\documentclass[11pt]{article}

\usepackage[numbers]{natbib}
\usepackage{geometry}
\usepackage{todonotes}
\usepackage[T1]{fontenc} 
\usepackage{libertinus}

\usepackage{microtype}
\DeclareMathAlphabet{\mathcal}{OMS}{zplm}{m}{n}

\usepackage[utf8]{inputenc}
\usepackage{amssymb}
\usepackage{amsmath}
\usepackage{amsthm}

\usepackage{nicefrac}
\usepackage{comment}
\usepackage[colorlinks]{hyperref}
\def\tmp#1#2#3{%
  \definecolor{Hy#1color}{#2}{#3}%
  \hypersetup{#1color=Hy#1color}}
\tmp{link}{HTML}{800006}
\tmp{cite}{HTML}{2E7E2A}
\tmp{file}{HTML}{131877}
\tmp{url} {HTML}{8A0087}
\tmp{menu}{HTML}{727500}
\tmp{run} {HTML}{137776}
\def\tmp#1#2{%
  \colorlet{Hy#1bordercolor}{Hy#1color#2}%
  \hypersetup{#1bordercolor=Hy#1bordercolor}}
\tmp{link}{!60!white}
\tmp{cite}{!60!white}
\tmp{file}{!60!white}
\tmp{url} {!60!white}
\tmp{menu}{!60!white}
\tmp{run} {!60!white}

\usepackage{cleveref} 

\usepackage{thmtools}
\usepackage{thm-restate}

\crefname{mechanism}{mechanism}{mechanisms}
\crefname{claim}{claim}{claims}

\theoremstyle{plain}
\newtheorem{theorem}{Theorem}[section]
\newtheorem{lemma}[theorem]{Lemma}
\newtheorem{corollary}[theorem]{Corollary}

\theoremstyle{definition}
\newtheorem{definition}[theorem]{Definition}

\theoremstyle{remark}

\usepackage{mathtools}

\usepackage{color}
\usepackage{xcolor}

\usepackage[shortlabels]{enumitem}
\usepackage[ruled,vlined,linesnumbered]{algorithm2e}

\DontPrintSemicolon

\IncMargin{0em}
\SetInd{0.5em}{0.7em}

\SetCommentSty{mycommfont}
\newcommand{\algcom}[1]{\tcp*[r]{#1}}
\newcommand{\algcomf}[1]{\tcp*[f]{#1}}
\SetAlFnt{\normalfont}

\let\oldnl\nl
\newcommand{\nonl}{\renewcommand{\nl}{\let\nl\oldnl}}

\makeatletter
\newenvironment{mechanism}[1][htb]{%
    
   \begin{algorithm}[#1]%
  }{\end{algorithm}}
\makeatother

\renewcommand{\vec}[1]{\mathbf{#1}}

\DeclareMathOperator*{\argmax}{arg\,max}

\newcommand{\set}[1]{\ensuremath{\{#1\}}}
\newcommand{\sset}[2]{\ensuremath{\{#1 \mid #2\}}}

\newcommand{\mech}{\ensuremath{\mathcal{M}}}
\newcommand{\opt}{\ensuremath{\textsc{opt}}}

\newcommand{\minisec}[1]{\bigskip\noindent\textbf{#1.~~}}

\newcommand{\tc}{\ensuremath{t}}
\newcommand{\dc}{\ensuremath{c}}
\newcommand{\gold}{\ensuremath{\vec{\dc}^{\textrm{GT}}}}
\newcommand{\wood}{\ensuremath{\vec{\dc}^{\textrm{WS}}}}

\newcommand{\apx}{\textsc{apx}}

\newcommand{\wwm}{\textsc{WillyWonka}}

\usepackage[noblocks]{authblk}

\title{ The Effectiveness of Golden Tickets and Wooden Spoons for Budget-Feasible Mechanisms}

\author[1]{Bart de Keijzer}
\affil[1]{Department of Informatics, King's College London, UK.}
\author[2,3]{Guido Sch{\"a}fer}
\affil[2]{Centrum Wiskunde \& Informatica (CWI), The Netherlands. }
\affil[3]{Institute for Logic, Language and Computation, University of Amsterdam, The Netherlands.  }
\author[2]{Artem Tsikiridis}
\author[1]{Carmine Ventre}

\date{}

\begin{document}

\maketitle
\begin{abstract}%
\noindent 
One of the main challenges in mechanism design is to carefully engineer incentives ensuring truthfulness while maintaining strong social welfare approximation guarantees. But these objectives are often in conflict, making it impossible to design effective mechanisms. An important class of mechanism design problems that belong to this category are \emph{budget-feasible mechanisms}, introduced by Singer (2010). Here, the designer needs to procure services of maximum value from a set of agents while being on a budget, i.e., having a limited budget to enforce truthfulness. It is known that no deterministic (or, randomized) budget-feasible (BF) mechanism satisfying dominant-strategy incentive compatibility (DSIC) can surpass an approximation ratio of $1+\sqrt{2}$ (respectively, $2$). However, as empirical studies suggest, factors like limited information and bounded rationality question the idealized assumption that the agents behave perfectly rationally. Motivated by this, Troyan and Morill (2022) introduced \emph{non-obvious manipulability (NOM)} as a more lenient incentive compatibility notion, which only guards against “obvious” misreports.

In this paper, we investigate whether resorting to NOM enables us to derive improved mechanisms in budget-feasible domains. We establish a tight bound of $2$ on the approximation guarantee of BF mechanisms satisfying NOM for the general class of monotone subadditive valuation functions. In terms of computational restrictions, the actual approximation ratio depends on the oracle model used to solve the respective allocation problem, but NOM imposes an impossibility barrier of $2$ only for deterministic mechanisms. Our result thus establishes a clear separation between the achievable guarantees for DSIC (perfectly rational agents) and NOM (imperfectly rational agents). En route, we fully characterize BNOM and WNOM (constituting NOM) and derive matching upper and lower bounds, respectively. Conceptually, our characterization results suggest \emph{Golden Tickets} and \emph{Wooden Spoons} as natural means to realize BNOM and WNOM, respectively. Equipped with these insights, we extend our results to more complex feasibility constraints and show that, basically, the same design template can be used to guarantee NOM. Additionally, we show that randomized BF mechanisms satisfying NOM can achieve an expected approximation ratio of $1+\varepsilon$ for any $\varepsilon > 0$. 
\end{abstract}

\section{Introduction}

Consider the problem of hiring a set $N = \set{1, \dots, n}$ of service provider agents, each with a provision cost $c_i \ge 0$ and a value $v_i \ge 0$ for being hired. Assume for simplicity that values are publicly known whilst costs are private knowledge of the agents. The tension between costs and values gives rise to the interesting problem of hiring a subset of agents of maximum total value (which can be thought of as the sum of the values of the hired agents) subject to feasibility constraints on the costs. More specifically, the designer needs to pay providers according to ``value for money'' whilst covering the agents' costs. Without further guarantees, however, it would be beneficial for the maximum-value agents, amongst potentially others, to over-report their costs in order to increase their profits. This is where mechanism design principles may help out: \emph{incentive compatibility} of the mechanism used guarantees that the agents will not attempt to misguide the auctioneer, who can then get as good an approximation as possible for the underlying knapsack procurement auction problem. 

But can the auctioneer afford to enforce incentive compatibility? If the latter is too expensive a goal, then its guarantee is only theoretical and not realistic in practice. To address this issue, \citet{singer10} introduced the concept of \emph{budget feasibility}. Not only do we want the mechanism to ensure that agents do not misbehave, but we also want the payments from the designer to the agents to be bounded by a given budget $B$. How effective can a solution be if we impose both desiderata? Research on this question has led to a series of works studying the design of budget-feasible mechanisms for various settings. In the most basic setting with additive valuation functions, \citet{gravin20} provide the current best $3$-approximate deterministic budget-feasible and incentive compatible mechanism for this problem. \citet{chen11} proved that the best possible approximation guarantee for deterministic budget-feasible mechanisms that are incentive compatible is $1 + \sqrt{2}$. It is known that randomized mechanisms can perform strictly better: \citet{gravin20} derive a randomized budget-feasible and incentive compatible mechanism that achieves an approximation ratio of $2$ (in expectation), which, as they show, is the best possible approximation for incentive compatible randomized mechanisms. It is worth mentioning that both of the aforementioned lower bounds hold independently of any computational constraints; the primary challenge is the requirement of incentive compatibility.

But are agents perfectly rational in practice? Empirical research shows that, on one end of the spectrum, people may misreport even when it is not beneficial to them \cite{kagel87}, whilst on the other, people may fail to strategize even when it would lead to a higher payoff \cite{troyan20}. These two observations have led to the alternative incentive compatibility notions of \emph{obvious strategyproofness (OSP)} \cite{li17} and \emph{non-obvious manipulability (NOM)} \cite{troyan20}, respectively. OSP makes sure that incentive compatibility is preserved even when imperfectly rational agents tend to strategize needlessly, and is a strengthening of the notion of incentive compatibility considered in earlier work. Adopting this pessimistic viewpoint can thus only worsen the achievable approximation guarantee of budget-feasible mechanisms. Whilst in general OSP comes at a cost in terms of approximation guarantee, see, e.g., \cite{ashlagigonczarowski,MOR22,ec2023}, 
budget feasibility represents an unexpected exception in that there is no separation between classic incentive compatibility and OSP, see e.g.,  \cite{balkanski22, trieagle}. NOM takes a more optimistic perspective and guarantees incentive compatibility when agents only decide to engage with misreports that are ``obviously beneficial'' to them, defined in a certain precise sense. This notion enlarges the class of incentive compatible mechanisms, and is backed by empirical evidence of behavioral biases. How much better can the approximation guarantee of budget-feasible NOM mechanisms be? The algorithmic power of this more permissive notion of incentive compatibility is the \emph{unique} driving force behind our investigations presented in this paper. Importantly, we neither endorse nor reject the idea that NOM is the `right' notion of incentive compatibility for imperfectly rational agents. We refer the interested reader to \cite{Troyan2022} for its rationale and extensive motivational applications. Our focus is purely algorithmic here: What is the difference between the algorithms of NOM mechanisms as opposed to (O)SP?

\minisec{Our Contributions and Technical Merits}
We initiate the study of budget-feasible mechanism design under the solution concept \emph{non-obvious manipulability (NOM)}. In a nutshell, we give an algorithmic recipe for NOM mechanisms and show its flexibility, providing strong guarantees in terms of approximation ratios and applicability to general mechanism design problem with feasibility constraints. More specifically, we identify three main contributions.

\begin{description} 
\item[Main Result 1:] 
We provide a full understanding of the algorithmic power of budget-feasible NOM mechanisms, covering both deterministic and randomized approaches.
\end{description}

Technically, NOM requires to compare certain extremes of the utility that agents can experience. Specifically, the minimum (maximum, respectively) utility of each agent $i$ over the reports of all other agents for truth-telling must be no lower than the minimum (maximum, respectively) utility for lying; this property is called \emph{worst-case NOM (WNOM)} (or, \emph{best-case NOM (BNOM)}, respectively). 
A budget-feasible mechanism is NOM if it satisfies both BNOM and WNOM. 
Recent work by \citet{archbold24} introduces a convenient technique for the design of NOM mechanisms, which are captured by a class of mechanisms that the authors dub \emph{Willy Wonka mechanisms}, to guarantee either property. In a Willy Wonka mechanism, BNOM is guaranteed by the definition of a \emph{golden ticket} as part of the mechanism, whereas WNOM is satisfied via a \emph{wooden spoon}. In our context, golden tickets (and wooden spoons, respectively) mean that for each agent $i$ and each cost declaration $c_i$, there exists a cost profile of the other agents for which $i$ is hired and paid the whole budget $B$ (and not hired, respectively). 

As our first technical contribution, we demonstrate how to leverage this framework to design a deterministic, budget-feasible mechanism that is NOM and $2$-approximate when the value of the hired set $S$ is determined by a monotone, subadditive function. This class of valuation functions is significantly more expressive than additive or submodular functions. The incentive properties of our mechanism come from the definitions of golden tickets and wooden spoons for cost profiles wherein either overpaying hired agents or forcibly not hiring agents does not impose too large a loss in terms of approximation. No budget-feasible deterministic strategyproof mechanism can achieve an approximation ratio better than $(1+\sqrt{2})$ (even for additive valuation functions). For context, the state of the art is $4.75$ for monotone submodular valuations with clock (OSP) auctions \cite{trieagle}. For the relaxed incentive compatibility notion of NOM, we achieve a strictly better approximation guarantee of $2$, albeit for a significantly larger class of valuation functions.
As a side result, we also demonstrate that our mechanism can handle more complex combinatorial domains. The respective approximation guarantee depends on what we coin the \emph{agent-enforcing gap} (see below for details).
The algorithmic power of NOM becomes even more striking when we are allowed to randomize over a set of deterministic NOM mechanisms. 
We present a simple randomized, budget-feasible mechanism that is universally NOM and achieves an approximation guarantee of $1+\varepsilon$ for any $\varepsilon > 0$ when valuation functions are monotone and subadditive. If computational constraints are a concern, the approximation ratio degrades to that of the best achievable guarantee for the underlying knapsack problem with monotone subadditive valuations.
For monotone submodular valuations, randomized clock auctions can get a $4.3$ approximation \cite{trieagle} whilst no randomized strategyproof mechanism can do better than $2$ \cite{chen11}. 
The significant gain that can be made with randomization sheds light on the algorithmic flexibility of the golden ticket and wooden spoon recipe. 
Strategyproofness requires a threshold for each agent $i$ and each $\vec{c}_{-i}$. This rigid structure constrains the approximation guarantee of algorithms forced to produce suboptimal solutions and severely limits the benefits of randomization.
On the contrary, the Willy Wonka framework is much more forgiving defining one BNOM and one WNOM threshold for each $i$. With a large enough support, these thresholds can be drawn in a way that the loss in approximation occurs with negligible probability leading to a much better upper bound.

Our second main contribution shows that Willy Wonka is the only design template for NOM, thus showing that the gain provided by our mechanisms is tight.
\begin{description}
\item[Main Result 2:] 
We fully characterize the class of budget-feasible mechanisms that satisfy NOM. We derive a characterization both for WNOM and BNOM independently.
\end{description}

We \emph{fully} characterize (direct revelation) NOM and individually rational (IR for short, guaranteeing that honest agents will never incur into a loss) mechanisms for binary allocation problems and single-dimensional agents. By `fully' here we mean that we characterize both the allocation/selection function and payment function of any NOM mechanism in our procurement setting. 
Our characterization refines the characterization given in \cite{AdKV2023a} for NOM mechanisms and general outcome spaces in that we not only explicitly describe the selection behaviour of NOM for binary allocation problems but also characterize the space of admissible payments. An instructive way to look at our characterization is in terms of the aforementioned \emph{thresholds}. As noted above, strategyproofness requires the existence of a threshold; the agent is hired (not hired, respectively) for each cost lower (higher, respectively) than the threshold; a hired agent is paid the threshold for the particular $\vec{c}_{-i}$. 
The way in which NOM relaxes strategyproofness is as follows. BNOM requires the existence of a threshold for each $i$ such that $i$ is \emph{never} hired for each $(c_i, \vec{c}_{-i})$ for $c_i$ bigger than the threshold, and hired in \emph{at least} one $(c'_i, \vec{c}_{-i})$ for each $c_i'$ lower than the threshold. When hired, the payment received by $i$ can never be higher than the threshold. WNOM instead requires a threshold that is in a sense the dual of BNOM's threshold: for agent $i$ bidding below it, $i$ is \emph{always} hired in each $(c_i, \vec{c}_{-i})$ and for agent $i$ bidding above it, $i$ is not hired in \emph{at least} one $(c'_i, \vec{c}_{-i})$. When hired, the agent receives a payment which is at least the value of the threshold. 

Armed with the characterization of WNOM and BNOM for our setting, we derive lower bounds on the achievable approximation guarantees of deterministic budget-feasible mechanisms, separately for BNOM and WNOM, for additive valuation functions.

\begin{description}

\item[Main Result 3:] 
We identify the optimal approximation factors of $2$ and $\varphi$ (the golden ratio), that can be achieved by an individually rational and budget feasible mechanism under BNOM and WNOM respectively. 
Additionally, we present a distinct mechanism for WNOM that achieves the optimal $\varphi$-approximation guarantee, demonstrating a non-trivial approach to this problem.
\end{description}

Consequently, no deterministic budget-feasible mechanisms satisfying NOM can achieve an approximation guarantee strictly better than $2$. It is BNOM (rather than WNOM) the barrier to improving further the approximation factor of our main NOM mechanism. In particular, this shows that our mechanism is optimal for all classes of valuation functions ranging from additive to monotone subadditive. We find that WNOM is in fact less demanding, drawing an interesting parallel with past findings about the bilateral trade problem \cite{AdKV2023a}.

\minisec{Related Work}
The NOM solution concept was introduced by \citet{troyan20} to model scenarios where agents, due to limited contingent reasoning abilities, might not fully exploit strategic behavior to their advantage. In their work, \citet{troyan20} characterize NOM mechanisms in both settings without monetary transfers (such as bipartite matching) and with monetary transfers (such as auctions and bilateral trade). Interestingly, for the bilateral trade setting, \citet{AdKV2023b} showed in a follow-up work a separation between WNOM and BNOM, which, in some sense, aligns with the results of our present work. Specifically, they demonstrate that while there exist bilateral trade mechanisms that are individually rational, efficient, weakly budget-balanced, and WNOM, the same cannot be achieved when substituting WNOM with BNOM under any relaxation of weak budget-balance. Finally, NOM has also been studied in various other domains, including  
voting rules \cite{AzizLam2021,elkindneohteh}, fair division \cite{PsomasVerma2022, OrtegaSegal-Halevi2022}, assignment mechanisms \cite{Troyan2022}, hedonic games \cite{flamminietal}, and settings 
beyond direct-revelation mechanisms \cite{AdKV2023a}.

The concepts of ``golden ticket'' and ``wooden spoon'' were introduced by \citet{archbold24}, who focus on designing prior-free mechanisms for single-parameter domains with the objective of attaining revenue guarantees. Additionally, \citet{archbold24} characterize NOM mechanisms for the settings they study in terms of their allocation functions \footnote{Note that our characterization of NOM budget-feasible mechanisms in Section \ref{sec:characterization} is in terms of both the allocation function and payment function.}. In particular, they show that all NOM mechanisms belong to a class which they call Willy Wonka Mechanisms.

As noted above, OSP takes the opposite perspective to bounded rationality asking for truthtelling to be an ``obviously'' dominant strategy. OSP mechanisms have been considered in various contexts; many results indicate that OSP is a demanding restriction of the classical notion of strategyproofness. For example, there is no OSP mechanism that returns a stable matching \cite{ashlagigonczarowski}. Moreover, for single-dimensional agent domains, there is a link between greedy algorithms and OSP \cite{wine21,ec2023} essentially indicating that the approximation quality achievable with these mechanisms is as good as that of 
a greedy algorithm.

The study of strategy-proof budget-feasible procurement auctions was initiated by \citet{singer10}. Beyond the canonical setting with additive valuations \citet{singer10} (see also \cite{gravin20, chen11}), several other budget-feasible mechanism design settings have been explored in the literature. These include settings with more general valuation functions such as monotone submodular valuation functions \cite{chen11, jalaly18, balkanski22, trieagle}, non-monotone submodular valuation functions \cite{amanatidis19, huang15, trieagle}, and XOS and subadditive valuation functions \cite{bei17, amanatidis17, dobzinski11, balkanski22, neogi24}. Additionally, different feasibility restrictions have been considered \cite{amanatidis16, huang15, leonardi17}. Finally, the model has been studied under a large-market assumption \cite{anari18} and for divisible agents \cite{klumper22, amanatidis23}. For further details, we refer the reader to the recent survey of \citet{liu24}.

\section{Model and Preliminaries} \label{sec:model}

We consider a procurement auction involving a set of agents $N = \{1, \dots, n\}$ and an auctioneer who has some budget $B \in \mathbb{R}_{> 0}$ available. Each agent $i \in N$ offers a service and has a private cost parameter $\tc_i \in \mathbb{R}_{\geq 0}$, representing their true cost for providing this service. The auctioneer has access to a set function $V:2^N \mapsto \mathbb{R}_{\geq 0}$, which determines their value of hiring any subset $S \subseteq N$ of agents. Throughout the paper, we assume that $V$ is \emph{monotone}, i.e., for every $S \subseteq T\subseteq N$, it holds that $V(S) \leq V(T)$, and \emph{normalized}, i.e., $V(\emptyset) = 0$.

A deterministic mechanism $\mathcal{M}$ consists of an allocation rule $\vec{x}: \mathbb{R}_{\geq 0}^n \to \{0,1\}^n$ and a payment rule $\vec{p}: \mathbb{R}_{\geq 0}^n \to \mathbb{R}_{\geq 0}^n$. To begin with, the auctioneer collects a profile $\vec{\dc} = (\dc_i)_{i \in N} \in \mathbb{R}_{\geq 0}^n$ of declared costs from the agents. Here, $\dc_i \in \mathbb{R}_{\geq 0}$ denotes the cost \emph{declared} by agent $i \in N$ and may differ from their true cost $t_i$. Given the declarations $\vec{\dc}$, the auctioneer determines an allocation $\vec{x}(\vec{\dc}) = (x_1(\vec{\dc}), \dots, x_n(\vec{\dc}))$, where $x_i(\vec{\dc}) \in \{0,1\}$ is the allocation decision for agent $i$, i.e., whether agent $i$ is hired or not. Given an allocation $\vec{x}(\vec{\dc})$, we use $X(\vec{\dc}) = \{i \in N \mid x_i(\vec{\dc}) = 1\}$ to refer to the set of agents who are selected under $\vec{x}(\vec{c})$; we use $\vec{x}(\vec{c})$ and $X(\vec{c})$ interchangeably. The auctioneer also determines a vector of payments $\vec{p}(\vec{\dc}) = (p_1(\vec{\dc}), \dots, p_n(\vec{\dc}))$, where $p_i(\vec{\dc}) \in \mathbb{R}_{\geq 0}$ is the payment that agent $i$ receives for their service. Overall, we denote an instance of our procurement auction environment by the tuple $I = (N, \vec{\dc}, V, B)$. When part of the input is clear from the context, we often refer to an instance by its declared cost profile $\vec{\dc}$ simply.

\minisec{Imperfect Rationality and Non-Obvious Manipulability}
We consider the setting where each agent $i \in N$ wants to maximize their utility function which is \emph{quasi-linear}, i.e., the utility of agent $i$ with true (private) cost $t_i \geq 0$ for a profile $\vec{\dc}$ is $u_i^{\tc_i}(\vec{\dc})= p_i(\vec{\dc})-t_i \cdot x_i(\vec{\dc})$.

Mechanisms that are individual rational and strategyproof for single parameter domains (as considered here) are generally well understood. \citet{myerson81} gives a complete characterization of the properties that an allocation rule must satisfy and provides a formula to derive the corresponding payments. In this work, we focus on a less stringent notion of incentive compatibility, called \emph{not obviously manipulable} (NOM), which is more suitable if the agents are imperfectly rational (see \cite{troyan20}).

\begin{definition}[\cite{troyan20}]\label{def:nom}
A mechanism $\mech = (\vec{x}, \vec{p})$ is \emph{not obviously manipulable (NOM)} if it satisfies the following two properties: 
\begin{itemize}\itemsep0pt

\item \emph{Best-Case Not Obviously Manipulable (BNOM):} for every agent $i \in N$, and all $t_i \in [0,B]$ it holds that
\begin{equation}\label{eq:BNOM}
\sup_{\vec{\dc}_{-i}} u_i^{t_i}(t_i, \vec{\dc}_{-i}) \ge 
\sup_{\vec{\dc}_{-i}} u_i^{t_i}(\dc_i, \vec{\dc}_{-i}) \qquad \forall \dc_i \ge 0.
\end{equation}

\item \emph{Worst-Case Not Obviously Manipulable (WNOM):} for every agent $i \in N$, and all $t_i \in [0,B]$ it holds that
\begin{equation}\label{eq:WNOM}
\inf_{\vec{\dc}_{-i}} u_i^{t_i}(t_i, \vec{\dc}_{-i}) \ge 
\inf_{\vec{\dc}_{-i}} u_i^{t_i}(\dc_i, \vec{\dc}_{-i}) \qquad \forall \dc_i \ge 0.
\end{equation}
\end{itemize}
\end{definition} 

\minisec{Additional Design Objectives}
In our procurement auction environment, we are interested in designing a not obviously manipulable mechanism $\mech = (\vec{x}, \vec{p})$ that
additionally satisfies the following properties for every given cost profile $\vec{\dc}$:
\begin{itemize}\itemsep0pt

\item \emph{Individual Rationality (IR):~~} 
An agent that is selected receives at least their declared cost as payment, and the payments are non-negative always, i.e., for every $i \in N$, $p_i(\vec{\dc}) \ge \dc_i \cdot x_i(\vec{\dc})$.

\item \emph{Normalized Payments (NP):~~} 
An agent that is not selected receives no payment, i.e., for every $i \notin X(\vec{\dc})$, $p_i(\vec{\dc})=0$.

\item \emph{Budget-Feasibility (BF):~~}
The auctioneer's total amount of payments made to the agents does not exceed their budget $B$, i.e., $\sum_{i \in N}p_i(\vec{\dc}) \le B$.
\end{itemize}

Note that a mechanism $\mech$ satisfying both IR and BF cannot hire any agent $i \in N$ whose declared cost is larger than the budget $B$, i.e., $\dc_i>B$. Throughout the paper, we assume without loss of generality that $\vec{\dc} \in [0,B]^n$ and $t_i \in [0,B]$ for all agents $i \in N$.\footnote{Note that this assumption is without loss of generality as we do not impose restrictions on the bidding space: given a cost profile $\vec{\dc}$, each agent $i \in N$ with $\dc_i>B$ can simply be discarded by the mechanism.}

\minisec{Valuation Functions and Approximation Guarantees} 
In budget-feasible mechanism design, the quality of the outcome computed by a mechanism $\mech = (\vec{x}, \vec{p})$ is assessed with respect to an optimal solution of the underlying packing problem (see, e.g., \citet{singer10}). Basically, the intuition here is that we measure the relative loss incurred by the mechanism in comparison to the best possible allocation in the non-strategic setting. 

More formally, given an instance $I=(N, \vec{\dc}, V, B)$, we measure the performance of a mechanism $\mech=(\vec{x}, \vec{p})$ by comparing $V(X(\vec{\dc}))$ with the optimal solution of the following \emph{packing problem}:
\begin{equation}\label{eq:opt-packing}
\max \ \ V(X) \quad \text{s.t.} \quad \sum_{i \in X} c_i x_i \leq B, \quad X \subseteq N.
\end{equation}
We use $X^*(\vec{\dc})$ to refer to an optimal (player-set) solution of the above problem; similarly; we use $\vec{x}^*(\vec{c})$ to denote the respective (binary) allocation. We say that a deterministic mechanism $\mech = (\vec{x}, \vec{p})$ is an \emph{$\alpha$-approximation mechanism} with $\alpha \geq 1$ if for every cost profile $\vec{\dc}$ it holds that $\alpha \cdot V(X(\vec{\dc})) \geq V(X^*(\vec{\dc}))$.

Our mechanism makes use of an approximation algorithm for the packing problem in \eqref{eq:opt-packing}. We use $\apx$ to refer to any such (deterministic) algorithm.\footnote{We assume that the algorithm always selects agents that declare a cost of $0$, which is without loss of generality.} 
Given an instance $I = (N, \vec{\dc}, V, B)$, we use $X_{\apx}(\vec{\dc})$ and $\vec{x}^*_{\apx}(\vec{\dc})$ to refer to the (player-set) solution and (binary) allocation, respectively, computed by $\apx$ on $I$. 
Throughout the paper, we assume that $\apx$ achieves an approximation guarantee of $\gamma \ge 1$, i.e., 
$\gamma \cdot V(X_{\apx}(\vec{\dc})) \geq V(X^*(\vec{\dc}))$.

Unless stated otherwise, we consider instances where the valuation function $V$ is monotone and normalized, without further assumptions on its form. In some cases, however, we consider valuation functions $V$ of certain forms. The relevant functions are defined as follows: 
\begin{enumerate}\itemsep0pt
    \item \textbf{Additive Valuation Function:~} For every $S, T \subseteq N$, the function satisfies $V(S \cup T) = V(S) + V(T)$. We denote the additive valuation function by $V_{\textsc{add}}$. Equivalently, $V_{\textsc{add}}$ can simply be represented by a value profile $\vec{v} = (v_i)_{i \in N} \in \mathbb{R}^n_{\ge 0}$ and defining $V_{\textsc{add}}(S) = \sum_{i \in S} v_i$ for every $S \subseteq N$. 
    Clearly, for $V = V_{\textsc{add}}$, the packing problem in \eqref{eq:opt-packing} is the classical \emph{Knapsack Problem}. Even though the problem is known to be NP-hard, it is well known that it admits an FPTAS, i.e., there exists a polynomial-time approximation algorithm with $\gamma = 1 + \varepsilon$, for an arbitrarily small $\varepsilon > 0$. 
    \item \textbf{Subadditive Valuation Function:~} For every $S, T \subseteq N$, the function satisfies $V(S \cup T) \leq V(S) + V(T)$. We denote the subadditive valuation function by $V_{\textsc{sub}}$. When $V = V_{\textsc{sub}}$, the packing problem in \eqref{eq:opt-packing} is notoriously hard in the general case. In particular, \citet{singer10} has shown that obtaining a $\gamma = o(n)$ requires an exponential number of queries to $V_{\textsc{sub}}$.\footnote{In fact, this is true even for the special case of monotone fractionally subadditive functions, see e.g., \cite{lehmann01}.} 
    However, under the assumption of having access to a stronger oracle which uses \emph{demand queries}\footnote{We refer the reader to \citet{badanidiyuru19} for a formal definition of demand queries.}, \citet{badanidiyuru19} have shown that there is an algorithm which attains $\gamma = 2 + \varepsilon$, for an arbitrarily small $\varepsilon > 0$, using a polynomial number of oracle calls. 

    \item \textbf{Submodular Valuation Functions:~}
    Finally, we note that the class of monotone subadditive functions includes the class of \emph{monotone submodular functions}, i.e., for every $S, T \subseteq N$, $V(S \cup T) \le V(S) + V(T) - V(S\cap T)$.
    For these functions there is an approximation algorithm that achieves an approximation guarantee of $\gamma = \frac{e}{e-1} \approx 1.58$, even for the standard model with value queries (see e.g., \citet{khuller99} and \citet{sviridenko04}).

\end{enumerate}

\section{A Framework to Design NOM Budget-Feasible Mechanisms}
\label{sec:mechanism-subadditive}

The main result of this section is a general-purpose deterministic mechanism that is individually rational (IR), budget-feasible (BF) and not obviously manipulable (NOM), named $\wwm$ (Mechanism \ref{mech:ww}). Even though this mechanism achieves all our mechanism design objectives, as a stand-alone mechanism it does not offer tangible approximation guarantees for the auctioneer. However, the idea is to use this mechanism in certain types of \emph{compositions} of mechanisms (both deterministic and randomized) that also perform well in terms of approximation guarantee.
More specifically, by combining our $\wwm$ mechanism with another simple mechanism, we derive a deterministic mechanism that is IR, BF and NOM and achieves an approximation guarantee of $2$ for the general class of monontone subadditive valuation functions.
The approximation guarantee of our mechanism is best possible: As we show in Section~\ref{sec:bnomchar}, even for the more restrictive class of additive valuation functions, no deterministic mechanism satisfying IR, BF and BNOM (only) can achieve an approximation guarantee strictly better than $2$.

The section is structured as follows. In Section \ref{subsec:willy-wonka}, we present our $\wwm$ mechanism and show that it satisfied IR, BF and NOM. Then, in Section \ref{subsec:MaxOrWW}, we combine our $\wwm$ mechanism with another simple mechanism (trivially satisfying IR, BF and NOM) and prove that the resulting mechanism achieves an approximation ratio of $2$ for monontone subadditive valuation functions.

\subsection{NOM Through Golden Tickets and Wooden Spoons}\label{subsec:willy-wonka}

We describe our new $\wwm$ mechanism in more detail. 
The mechanism achieves NOM by implementing \emph{golden tickets} and \emph{wooden spoons}, which is a technique that was first introduced by \citet{archbold24}. Here, we adapt this technique to our budget-feasible mechanism design setting. 
As it turns out, 
it suffices to use simple, but carefully designed golden tickets and wooden spoons. In Section~\ref{sec:characterization}, we introduce much more refined notions of golden tickets and wooden spoons in order to derive our characterization results for BNOM and WNOM. 

\minisec{Golden Tickets and Wooden Spoons Technique}
The high-level idea behind this technique is do define for each agent two special cost profiles of the opposing agents, called \emph{golden ticket} and \emph{wooden spoon}, that trigger particular outcomes when they occur. 

The golden tickets realize BNOM by ensuring that for each agent $i \in N$ the left-hand side of the BNOM condition \eqref{eq:BNOM} attains maximum utility. More specifically, for each agent $i$ and cost $c_i \in [0, B)$, the golden ticket $\gold_{-i}$ is a cost profile of the opposing agents for which $i$ is chosen and paid the maximum amount $B$. 
That is, if $i$ declares their private type $c_i = t_i$, the golden ticket ensures that the utility of $i$ is $u_i^{t_i}(t_i, \gold_{-i}) = B - t_i$. Thus, the supremum on the left-hand side of \eqref{eq:BNOM} attains the maximum possible utility, establishing BNOM. 

Similarly, the wooden spoons implement WNOM by ensuring that for each agent $i \in N$ the right-hand side of \eqref{eq:WNOM} is non-positive. More concretely, for each agent $i$ and cost declaration $c_i \in [0, B)$, the wooden spoon $\wood_{-i}$ is a cost profile of the opposing agents for which $i$ is rejected and paid zero (or $i$ is accepted and paid zero). In particular, for any cost declaration $c_i$, the wooden spoon ensures that the utility of $i$ is $u_i^{t_i}(c_i, \wood_{-i}) \le 0$. 
The infimum on the right-hand side of \eqref{eq:WNOM} is thus non-positive, proving WNOM. 

It is crucial that the golden tickets and wooden spoons are defined such that they do not `interfere' with each other, i.e., each agent $i$ must be able to implement their golden ticket and wooden spoon. 
Said differently, for a given cost profile $\vec{c}$, if $i$ admits their golden ticket, i.e., $\vec{c}_{-i} = \gold_{-i}$, or wooden spoon, i.e., $\vec{c}_{-i} = \wood_{-i}$, then the respective allocation and payment must be effectuated for $i$.

\begin{table}[t]
    \centering
    \begin{tabular}{|@{\qquad}c@{\qquad}|@{\qquad}c@{\qquad}c@{\qquad}|}
    \hline 
        & & \\[-2ex]
        agent $i$ & $(c_i,\, \gold_{-i})$ & $(c_i,\, \wood_{-i})$ \\[.5ex]
        \hline\hline
        & & \\[-2ex]
         1 & $(c_1, B, B)$ & $(c_1, 0, 0)$ \\[0ex]
         2 & $(0, c_2, B)$ & $(0, c_2, 0)$ \\[0ex]
         3 & $(0, 0, c_3)$ & $(B, B, c_3)$ \\[.5ex]
        \hline
    \end{tabular}
    \caption{Golden tickets and wooden spoons for $n = 3$. 
    Note that these profiles are not interfering for $c_i \neq B$. The only non-unique cost profile is 
    $\vec{c} = (0, 0, 0)$ which is the golden ticket of agent 3 and the wooden spoons of agents 1 and 2. 
    But these are effectuated by choosing all three agents and paying $B$ to agent $1$ and $0$ to agents 2 and 3.  
    Note also that the restriction $c_i \neq B$ is crucial as otherwise the golden tickets would be interfering. For example, then $(0, B, B)$ would be the golden ticket of agent 1 and 2. But these cannot be effectuated without violating budget feasibility (as both agents would have to be selected and paid $B$).}
    \label{tab:GT-WS}
\end{table}

\minisec{Golden Tickets and Wooden Spoons for \wwm}
The golden tickets and wooden spoons used by our $\wwm$ mechanism are defined as follows. For each agent $i \in N$ with $c_i \in [0, B)$, we define: 
\begin{equation}
\gold_{-i} = (\underset{1}{0}, \underset{\dots}{\dots}, \underset{i-1}{0}, \underset{i+1}{B}, \underset{\dots}{\dots}, \underset{n}{B}) 
\qquad \text{and}\qquad
\wood_{-i} = 
\begin{cases} (0, \ldots, 0) & \text{ if } i < n \\
(B, \ldots, B) & \text{ if }i = n .
\end{cases}
\label{eq:gtws}
\end{equation}

It is not hard to verify that these golden tickets and wooden spoons are not interfering with each other. 
In fact, for this it is crucial that they are defined for $c_i \neq B$ only. It will become clear below that the case $c_i = B$ is handled automatically (due to IR and BF of the mechanism). 
The golden tickets and wooden spoons for $n = 3$ are given in Table~\ref{tab:GT-WS}.

\begin{mechanism}[t]
\caption{$\wwm(I)$}\label{mech:ww}

\nonl \hspace*{-1em} $\rhd$ {\bf{Input:~~}} instance $I=(N, \vec{c}, V, B)$

Rename the agents so that $V(\{1\}) \ge V(\{2\}) \ge \dots \ge V(\{n\})$\label{alg:rename} \;
\label{alg:ww:else}
\If(\algcomf{$i$ gets golden ticket}){there is an $i \in N$ such that  $\dc_i \in [0, B)$ and $\vec{\dc}_{-i} = \gold_{-i}$ \label{alg:gt}}{%
$x_i = 1$, $p_i = B$ \label{alg:gti}
\\
$x_j = 1$, $p_j = 0$ for $j=1,\dots,i-1$ \label{alg:gtii} \\ 
$x_j = 0$, $p_j = 0$ for all $j =i+1,\dots, n$ \label{alg:ww:gt}\;
}
\ElseIf(\algcomf{$i$ gets wooden spoon}){there is an $i \in N$ such that $\dc_i \in [0, B)$ and $\vec{\dc}_{-i} = \wood_{-i}$ \label{alg:ws}}{
$x_i = 0$, $p_i = 0$ \label{alg:wsi}\\
\If(\algcomf{$\wood_{-i} = (B,\ldots,B)$}){i = n}{
$x_1 = 1$, $p_1 = B$ \label{alg:ws-p1}\\
$x_j = 0$, $p_j = 0$ for $j \in N\setminus\{1,n\}$ \label{alg:ws-p2}
}
\Else(\algcomf{$\wood_{-i} = (0,\ldots,0)$}){$x_j = 1$, $p_j = 0$ for all $j \in N \setminus \{i\}$} 
\label{alg:wsend} }
\Else{
$\vec{x} = \vec{x}_{\apx}(N, \vec{c}, V, B)$ \algcom{use $\apx$ to solve packing problem in \eqref{eq:opt-packing}} \label{alg:opt}
$\vec{p} = \vec{\dc \cdot \vec{x}}$ \label{alg:opt-payment} \algcom{pay-as-bid} \label{alg:ww:pab}
}
\Return{$(\vec{x}, \vec{p})$}
\end{mechanism}

\minisec{\wwm\ Mechanism}
A detailed description of our $\wwm$ mechanism is given in Mechanism~\ref{mech:ww}. 
The mechanism takes an instance $I = (N, \vec{c}, V, B)$ as input. 
It first renames the agents such that $V(\{1\}) \ge \dots \ge V(\{n\})$. Then, the mechanism checks whether any agent admits their golden ticket with respect to the given cost profile. 
Note that the all-zero cost profile $\vec{c} = (0, \dots, 0)$ is handled through the golden ticket $\gold_{-n}$ of agent $n$ with $c_n = 0$; as we will show below, this also implements the wooden spoon of each agent $i \in \{1, \dots, n-1\}$ with $c_i = 0$ correctly (even though $i$ is selected in this case). 
After that, the mechanisms verifies if any agent admits their wooden spoon. Note that for $i = n$ the wooden spoon $\wood_{-n} = (B, \dots, B)$ (which is structurally different to avoid interference) is implemented differently: agent $1$ is allocated and paid $B$, while all other agents are not allocated. 
Finally, if no golden tickets or wooden spoons apply, the mechanism uses a $\gamma$-approximation algorithm $\apx$ to compute an approximate solution $\vec{x}_{\apx}$ to the respective packing problem. In this case, the allocation is determined by $\vec{x}_{\apx}$ and each allocated agent is paid their cost.
Given an instance $I=(N, \vec{c}, V, B)$, we use $(\vec{x}(\vec{c}), \vec{p}(\vec{c}))=\wwm(I)$ to refer to the allocation and payments output by the mechanism. 

\minisec{Analysis}
We show that $\wwm$ satisfies all mechanism design objectives introduced in Section \ref{sec:model}. It is easy to verify that $\wwm$ is individually rational and budget-feasible.

\begin{lemma}\label{lem:ww-bf-ir}
    $\wwm$ is individually rational and budget-feasible.
\end{lemma}
\begin{proof}[Proof (Lemma~\ref{lem:ww-bf-ir})]
Let $(\vec{x}(\vec{c}), \vec{p}(\vec{c}))$ be the outcome computed by $\wwm$ on input $I=(N, \vec{c}, V, B)$.
It is easy to verify that IR holds by construction: 
For each agent $i \in X(\vec{\dc})$ the payment $p_i(\vec{c})$ is either $B \geq c_i$ (Lines \ref{alg:gti} and \ref{alg:ws-p1}) or $c_i$ (Lines \ref{alg:gtii}, \ref{alg:wsend} and \ref{alg:ww:pab}). 
Also, for each agent $i \notin X(\vec{\dc})$ we have $p_i(\vec{c}) = 0$ (Lines~\ref{alg:ww:gt}, \ref{alg:wsi} and \ref{alg:ws-p2}). 

Similarly, budget-feasibility is guaranteed because either a single agent is selected and paid $B$ (Lines \ref{alg:gti}--\ref{alg:ww:gt} and Lines \ref{alg:wsi}, \ref{alg:ws-p1} \& \ref{alg:ws-p2}), or it holds that $\sum_{i=1}^n p_i(\vec{c}) = \sum_{i=1}^n c_i \leq B$ because the approximation algorithm $\apx$ used in Line~\ref{alg:opt} computes a feasible solution to the packing problem. 
\end{proof}
We now show that $\wwm$ is not obviously manipulable.
\begin{lemma}\label{lem:ww-nom-bf-ir}
    $\wwm$ is not obviously manipulable.
\end{lemma}
\begin{proof}
Let $(\vec{x}(\vec{c}), \vec{p}(\vec{c}))$ be the outcome computed by $\wwm$ on input $I=(N, \vec{c}, V, B)$.
Consider an arbitrary agent $i \in N$. 
It suffices to prove that for all $t_i \in [0,B]$ the BNOM condition in (\ref{eq:BNOM}) and the WNOM condition in (\ref{eq:WNOM}) are satisfied.

First consider the case $t_i = B$. 
Then \eqref{eq:BNOM} and \eqref{eq:WNOM} follow directly from IR and BF: 
If $i$ declares cost $t_i = B$ then their utility is $u_i^{t_i}(t_i, \vec{c}_{-i}) = 0$ (either $i$ is selected and paid $B$ by IR, or $i$ is not selected).
On the other hand, if $i$ declares cost $c_i \in [0,B]$ we have $u_i^{t_i}(c_i, \vec{c}_{-i}) \le 0$ (either $i$ is selected and paid at most $B$ by BF, or $i$ is not selected). Thus,  \eqref{eq:BNOM} and \eqref{eq:WNOM} hold.

Let $t_i \in [0,B)$ and consider BNOM first. 
For every $t_i \in [0, B)$, there is a unique golden ticket $\gold_{-i}$ for $i$ (as defined in \eqref{eq:gtws}) that is implemented in Lines~\ref{alg:gti}--\ref{alg:ww:gt}: $i$ is selected and paid $B$. Thus, we have $u_i^{t_i}(t_i,\gold_{-i}) = B-t_i$. 
Since $u_i^{t_i}(\cdot)$ nowhere exceeds $B-t_i$, we have
$\sup_{\vec{\dc}_{-i}} u_i^{t_i}(t_i,\vec{\dc}_{-i}) = B-t_i \geq \sup_{\vec{\dc}_{-i}} u_i^{t_i}(c_i,\vec{\dc}_{-i})$ for all $c_i \in [0,B]$.
We conclude that (\ref{eq:BNOM}) holds.

Let $t_i \in [0, B)$ and consider WNOM. 
Observe that for all $c_i \in (0,B)$, there is a unique wooden spoon $\wood_{-i}$ for $i$ (as defined in \eqref{eq:gtws}) that is implemented in Lines~\ref{alg:wsi}--\ref{alg:wsend}: $i$ is not selected and paid $0$. 
The same holds if $c_i = 0$ and $i = n$. 
Thus, we have $u_i^{t_i}(c_i, \wood_{-i}) = 0$ in both cases.
If $c_i = 0$ and $i \neq n$, $\vec{c} = (c_i, \wood_{-i})$ is the all-zero profile which coincides with the golden ticket of agent $n$ with $c_n = 0$, which is implemented instead: agent $i \neq n$ is selected in this case but paid 0 (Line \ref{alg:gtii}).
Thus, we have $u_i^{t_i}(c_i, \wood_{-i}) \le 0$ in this case. 
Since $u_i^{t_i}(t_i, \cdot)$ is non-negative always, we have 
$\inf_{\vec{\dc}_{-i}} u_i^{t_i}(t_i,\vec{\dc}_{-i}) \ge 0 \ge u_i^{t_i}(c_i,\wood_{-i}) \geq \inf_{\vec{\dc}_{-i}} u_i^{t_i}(c_i,\vec{\dc}_{-i})$ for all $c_i \in [0,B]$. Thus (\ref{eq:WNOM}) holds.
\end{proof}

\subsection{Approximation Mechanism for Subadditive Valuations}
\label{subsec:MaxOrWW}

In this section, we derive a mechanism that is IR, BF and NOM and achieves an apporoximation guarantee of $\max\{2, \gamma\}$ for the general class of subadditive valuation functions. 

\minisec{Composed Mechanism}
The core idea underling our mechanism is as follows: The mechanism first checks whether there is an agent $i^{\star}$ who is ``valuable enough'' to be selected on their own, roughly compared to the optimal total value that all other agents can generate. If this is the case, the mechanism selects agent $i^{\star}$ and pays the entire budget $B$ (regardless of the declared cost $c_{i^{\star}} \in [0, B]$). Otherwise, it calls the $\wwm$ introduced above (see Mechanism \ref{mech:ww}). As it turns out, this composition allows us to prove attractive approximation guarantees.\footnote{Note that running either one of these two mechanisms alone does not provide any non-trivial approximation guarantee.}
The resulting mechanism is referred to as $\textsc{MaxOr}\wwm$ and given in Mechanism~\ref{mech:maxorww}. 

\begin{mechanism}[t]
\caption{$\textsc{MaxOr}\wwm(I)$}\label{mech:maxorww}

\nonl \hspace*{-1em} $\rhd$ {\bf{Input:}~~} instance $I=(N, \vec{c}, V_{\textsc{sub}}, B)$

Let $i^{\star}= \argmax_{i \in N}\frac{V(\{i\})}{V(N \setminus \{i\})}$\label{line:istar}

\If{$V(\{i^{\star}\}) \geq V(N \setminus \{i^{\star}\})$}{\label{alg:maxorww:if}
$x_{i^{\star}}=1,p_{i^{\star}}=B$\\
$x_i=0, p_i=0$ for all $i \in N \setminus \{i^{\star}\}$\\
}
\lElse{
    $(\vec{x}, \vec{p}) = \wwm(I)$
}
\Return{$(\vec{x}, \vec{p})$}
\end{mechanism}

\begin{restatable}{theorem}{thmwwguarantee}
\label{thm:ww-guarantee}
$\textsc{MaxOr}\wwm$ is an individual rational, budget-feasible and not obviously manipulable mechanism that achieves a $\max\{2,\gamma\}$-approximation guarantee for subadditive valuation functions, where $\gamma$ is the approximation guarantee of the algorithm $\apx$ used in $\wwm$. 
\end{restatable}

\begin{proof}
   Let $I=(N, \vec{c}, V_{\textsc{sub}}, B)$ be an instance with subadditive valuations $V_{\textsc{sub}}$. For brevity, let $V = V_{\textsc{sub}}$. 

    Note that whether the mechanisms runs the mechanism in the \textbf{if}-part or $\wwm(I)$ does not depend on the declared costs. In the former case, IR and BF hold by construction. Also, the utility of each agent is constant in this case and hence the mechanism is NOM. In the latter case, IR, BF and NOM are inherited from the $\wwm$ mechanism, as proven in Lemma \ref{lem:ww-nom-bf-ir}. It remains to prove that the approximation guarantee is $\max\{2, \gamma\}$.
    
   First, consider the case that $V(\{i^{\star}\}) \geq V(N \setminus \{i^{\star}\})$. 
   Then $X(\vec{c})=\{i^{\star}\}$. By the monotonicity and subadditivity of $V$, we have that $V(X^{*}(\vec{c})) \leq V(N) \leq V(N \setminus \{i^{\star}\}) + V(\{i^{\star}\}) \leq 2 \cdot V(\{i^{\star}\}) = 2 \cdot V(X(\vec{c})).$ 
Thus, the mechanism achieves a $2$-approximation in this case. 

Consider the case $V(\{i^{\star}\}) < V(N \setminus \{i^{\star}\})$. Then the $\wwm$ mechanism is run on $I$. Let $(\vec{x}(\vec{c}), \vec{p}(\vec{c}))$ be the output computed by $\wwm(I)$.
Below, all line numbers refer to $\wwm$.
We distinguish the following cases based on the profile $\vec{c}$.

\smallskip
\textbf{Case 1.~}
    There is an agent $i \in N$ such that $c_i \in [0,B)$ and $\vec{\dc}_{-i} = \gold_{-i}$. The outcome of the mechanism is determined by Lines \ref{alg:gti}--\ref{alg:ww:gt}, i.e., 
    $X(\vec{c})=\{1,\dots, i\}$.
    Note that, by the definition of $\gold_{-i}$, the optimal allocation $X^*(\vec{c})$ contains at most one agent $j \in \{i+1,\dots, n\}$, i.e., $X^*(\vec{c}) \subseteq \{1, \dots, i\} \cup \{j\}$. We thus obtain $V(X^*(\vec{c})) \leq V(\{1,\dots, i\} \cup \{j\}) \leq V(\{1,\dots, i\}) + V(\{j\}) \leq V(\{1,\dots, i\}) + V(\{i\}) \leq 2 \cdot V(\{1 ,\dots, i\})= 2 \cdot V(X(\vec{c})).$ 
    The first and last inequality hold by the monotonicity of $V$, the second inequality by the subadditivity of $V$, and the third inequality by the ordering of the agents on Line \ref{alg:rename} of $\wwm$. 

\smallskip
\textbf{Case 2.~} There is an agent $i \in N$ such that $c_i \in [0, B)$ and $\vec{\dc}_{-i} = \wood_{-i}$. 
    If $c_i = 0$ and $i \neq n$, then $\vec{c} = (c_i, \wood_{-i})$ is the all-zero profile which coincides with the golden ticket of agent $n$ with $c_n = 0$. In this case, the output is determined by the golden ticket of agent $n$ which has been analyzed in Case 1 above. 
    Otherwise, the outcome of the mechanism is determined by handling the wooden spoon of $i$ in 
    Lines~\ref{alg:wsi}--\ref{alg:wsend}.
    
    If $i=n$, by the definition of $\wood_{-n}$, the optimal allocation $X^*(\vec{c})$ contains at most one agent $j \in \{1,\dots, n-1\}$. 
    In addition, when $c_n=0$ it also contains agent $n$ due to the monotonicity of $V$. On the other hand, $X(\vec{c})=\{1\}$. Therefore, $V(X^*(\vec{c})) \leq V(\{j,n\}) \leq V(\{j\}) + V(\{n\}) \leq 2\cdot V(\{1\})=2 \cdot V(X(\vec{c})).$ 
    The second inequality follows from the subadditivity of $V$ and the third inequality by the ordering of agents on Line \ref{alg:rename} of $\wwm$.
    
    If $i \not= n$, then $X(\vec{c})=N \setminus \{i\}$ and $X^*(\vec{c}) = N$. Clearly, $V(X^*(\vec{c})) = V(N) \leq V(N \setminus \{i\}) + V(\{i\}) = V(N \setminus \{i\}) \cdot (1 + V(\{i\})/V(N \setminus \{i\}) ) \leq V(N \setminus \{i\}) \cdot (1 + V(\{i^{\star}\})/V(N \setminus \{i^{\star}\})) \leq 2 \cdot V(N \setminus \{i\})=2 \cdot V(X(\vec{c})).$ 
    
    The first inequality follows from the subadditivity of $V$, the second inequality from the definition of $i^{\star}$ in Line~\ref{line:istar} of $\textsc{MaxOr}\wwm$, and the last inequality holds because $V(\{i^{\star}\}) < V(N \setminus \{i^{\star}\})$ by assumption. 

 \smallskip
\textbf{Case 3.~} If none of the above cases hold, the outcome of the mechanism is determined by Lines \ref{alg:opt}--\ref{alg:ww:pab}, and thus
    $V(X^*(\vec{c})) \leq \gamma \cdot V(X_{\apx}(\vec{c})) = \gamma \cdot V(X(\vec{c}))$.
    
This concludes the proof that the approximation guarantee of the mechanism is $\max\{2, \gamma\}$. 
\end{proof}

\vspace*{-.3cm}

\minisec{Computational constraints}
We elaborate on the the trade-off that our $\textsc{MaxOr}\wwm$ mechanism implies in terms of achievable approximation guarantees versus computational efficiency. Setting computational considerations aside, we may assume access to an algorithm $\apx$ that solves the packing problem in \eqref{eq:opt-packing} optimally for monotone subadditive functions. In this scenario, $\gamma=1$, and therefore $\textsc{MaxOr}\wwm$ achieves an approximation ratio of $\max\set{2, \gamma} = 2$. By combining this fact with the best-known lower bound of $1 + \sqrt{2} \approx 2.41$ for deterministic budget-feasible mechanisms that are  individually rational and strategyproof (due to \citet{chen11}), we obtain a separation result in terms of achievable approximation guarantees between the classes of strategyproof and not obviously manipulable mechanisms. Note that the result of \citet{chen11} also holds without computational constraints and is for additive valuations. In conclusion, we have shown that not obviously manipulable mechanisms perform strictly better in terms of approximation than their strategyproof counterparts.

When polynomial running time of the mechanism is a desideratum, a negative result by \citet{singer10} implies that for monotone fractionally subadditive valuation functions (being a subclass of monotone subadditive) $\gamma = \Omega(n)$, as otherwise an exponential number of queries to $V$ would be required. In particular, this implies that the approximation guarantee of $\textsc{MaxOr}\wwm$ becomes $\Omega(n)$ as well.\footnote{Note that always choosing the agent of maximum singleton valuation gives a trivial $n$-approximation mechanism.} 
However, under the assumption of having access to a stronger demand oracle, the work of \citet{badanidiyuru19} provides an approximation algorithm with $\gamma = 2 + \varepsilon$ for arbitrarily small $\varepsilon > 0$ for the packing problem in \eqref{eq:opt-packing} with a monotone subadditive valuation function. That is, in this case $\textsc{MaxOr}\wwm$ achieves an approximation factor of $2 + \varepsilon$. Finally, the class of monotone subadditive functions includes the monotone submodular functions for which there is an approximation algorithm with $\gamma = \nicefrac{e}{e-1} \approx 1.58$. As a result, $\textsc{MaxOr}\wwm$ achieves an approximation factor of $2$ for monotone submodular valuations.

\minisec{Extensions to more complex feasibility constraints}
We argue that our design template of using golden tickets and wooden spoons, as defined in our WillyWonka mechanism, is versatile enough to handle more complex packing problems.
Suppose that instead of the simple packing problem in \eqref{eq:opt-packing-ext}, we consider the following \emph{general packing problem}:
\begin{equation}\label{eq:opt-packing-ext}
\max \ \ V(X) \quad \text{s.t.} \quad \sum_{i \in X} c_i x_i \leq B, \quad X \in \mathcal{F}(N). 
\end{equation}
Here $\mathcal{F}\subseteq 2^N$ may encode arbitrary restriction. 
Note that this captures our original problem if $\mathcal{F}(N) = 2^N$. But other feasibility restrictions are conceivable of course. For example, $\mathcal{F}(N)$ could be used to model pairwise conflicts (or dependencies) among the agents. 

The question we are asking here is the following one: Leaving computational restrictions aside, which properties of $\mathcal{F}$ ensure that our mechanism $\textsc{MaxOr}\wwm$ goes through as before, possibly providing an inferior approximation guarantee? 
A moment’s thought reveals that all that is required is verifying how the structure of $\mathcal{F}$ impacts the approximation guarantee obtained by issuing our golden tickets or wooden spoons.

Define the following allocations with respect to some agent set $S \subseteq N$: 
\begin{align*}
    \opt(S) & = \arg\max \sset{V(X)}{X \subseteq S, X \in \mathcal{F}}.\\ 
    \opt_{+i}(S) & = \arg\max \sset{V(X)}{X \subseteq S, X \in \mathcal{F}, i \in X}. \\
    \opt_{-i}(S) & = \arg\max \sset{V(X)}{X \subseteq S, X \in \mathcal{F}, i \notin X}.
\end{align*}

Consider the golden ticket $\vec{c} = (c_i, \gold_{-i})$ of agent $i$. 
Then, when issuing the golden ticket for agent $i$, we choose the optimal allocation $X(\vec{c}) = \opt_{+i}([i])$ forcing $i$ in and pay $B$ to agent $i$ as before; possibly $X(\vec{c})$ contains agent $i$ only.
Note that $X^*(c) \subseteq [i] \cup \set{j}$ for some $j > i$ (if any). 
Thus, exploiting monotonicity and subadditivity of $V$, we obtain 
\[
V(X^*(\vec{c})) \le V(X^*(\vec{c}) \cap [i]) + V(\set{j}) 
\le \delta V(X(\vec{c})) + V(\set{i})
\le (\delta + 1) V(X(\vec{c})).
\]
Here the inquality holds if we define the \emph{agent-forcing gap $\delta \ge 1$} such that 
\[
\forall S \subseteq N, \ \forall i \in S: \ \opt_{+i}(S) \ge \frac{1}{\delta} \opt(S). 
\]
Intuitively, the agent-forcing gap measures how much we lose in the worst case by forcing a single agent $i$ into the solution. 
Note that for our original model with $\mathcal{F}(N) = 2^N$, we have $\delta = 1$. 

Next, consider the wooden spoon $\vec{c} = (c_i, \wood_{-i})$ of agent $i$. When issuing the wooden spoon for agent $i$, we choose the optimal allocation $X(\vec{c}) = \opt_{-i}(N)$ forcing $i$ out and define the payments as before. 
Going through the same analysis as in the proof above, it turns out that in the worst case we are losing a factor of $2$ in the approximation guarantee (as before).

The only additional change we have to implement is the definition of the agent choice $i^{\star}$: Let $i^{\star} = \arg\max_{i \in N} \frac{V(\set{i})}{V(\opt_{-i})}$. With this, we are guaranteed that if $i^{\star}$ is output alone, the approximation guarantee is still at most $2$. 

We conclude that the adapted 
$\textsc{MaxOr}\wwm$ mechanism has a final approximation guarantee of $\max(2, \delta + 1, \gamma)$, where $\delta$ is the agent-forcing gap as defined above and $\gamma$ refers to the approximation ratio of the algorithm used to solve the underlying packing problem.

\section{Characterization of NOM and Impossibility Results} 
\label{sec:characterization}

In this section, we consider BNOM and WNOM separately, and provide characterizations for the class of IR, NP and BNOM mechanisms (Section \ref{sec:bnomchar}), as well as the IR, NP and WNOM mechanisms (Section \ref{sec:wnomchar}). Using our BNOM characterization, we show that the approximation factor that \textsc{MaxOr}\wwm\ (Mechanism \ref{mech:maxorww}) achieves cannot be improved, even for the class of additive valuations. For our WNOM characterization we establish a weaker lower bound of $\varphi = (1 + \sqrt{5})/2$ (golden ratio) on the achievable approximation factor for additive valuation, which we show to be the best possible for $3$ agents, by designing a mechanism for this case that has a matching approximation factor.

\subsection{BNOM Characterization and Approximation Bound of 2}\label{sec:bnomchar}

We provide a sufficient and necessary condition for mechanisms satisfying BNOM, provided that they satisfy the NP and IR properties. In words, the property states that for each agent $i \in N$ there is a particular threshold $b_i \in [0,B]$ such that: (1)    declaring a type strictly above $b_i$ guarantees $i$ to not get selected,
(2)    when declaring any type strictly below $b_i$, there exists a bid profile of the remaining agents under which $i$ gets selected and receives a payment of $b_i$,
(3)    when declaring type $b_i$, either of the above two cases apply.
We formalize the above as follows.
\begin{definition}\label{def:threshold-gt-payments}
    A mechanism $\mech$ uses \emph{threshold golden tickets} iff 
    \begin{align}
    \forall i \in N : \exists b_i \in [0,B] : & \quad \left(\rule{0ex}{3ex} \forall c_i \in (b_i,B] : \sup_{\vec{c}_{-i}} x_i(c_i,\vec{c}_{-i}) = 0\right) \wedge 
    \left(\forall c_i \in [0,b_i) : \sup_{\vec{c}_{-i}} p_i(c_i,\vec{c}_{-i}) = b_i \right) \notag \\
        & \, \wedge \left(\sup_{\vec{c}_{-i}} p_i(b_i,\vec{c}_{-i}) = b_i \vee \sup_{\vec{c}_{-i}} x_i(b_i,\vec{c}_{-i}) = 0\right) \label{eq:threshold-gt-payments} .
\end{align}
\end{definition}

\begin{restatable}{proposition}{propbnomchar}
\label{prop:bnomchar}
   A mechanism $\mech$ that satisfies NP and IR is BNOM if and only if $\mech$ uses threshold golden tickets.
\end{restatable}

It can be seen that \textsc{MaxOr}\wwm~(Mechanism \ref{mech:maxorww}), which is NOM, and hence BNOM, uses threshold golden tickets with $b_i=B$ for all $i \in N$: If $i$ declares anything less than $B$, then $i$ receives a payment of $B$ when the profile of the other bidders is $\vec{c}_{-i}^{GT}$. Furthermore, if $i$ declares $B$, then by IR the payment to $i$ is at least $B$, in case $i$ gets selected.

The proof of Proposition \ref{prop:bnomchar} uses an intermediate characterization of BNOM, for the wider class of mechanisms that satisfy NP (and not necessarily IR). The latter characterisation consists of two parts: 
\begin{itemize}\itemsep0pt
    \item If there exists a type $c_i \in [0,B]$ for Agent $i \in N$ such that $i$ is guaranteed to not get selected when declaring $c_i$ (i.e., regardless of the declarations of the other agents), then the maximum payment that the mechanism can award to $i$ across all profiles is $c_i$.
    \item For every type $c_i \in [0,B]$ of Agent $i$ such that $i$ \emph{can} get selected when declaring $c_i$, the highest payment $i$ can receive when declaring $c_i$ is equal to the highest payment that the mechanism can give to $i$ across all  profiles.
\end{itemize} 
We formalise the above as follows.
\begin{definition}\label{def:restricted-gt-payments}
A mechanism $\mech$ uses \emph{restricted 
 golden ticket payments} iff
\begin{align}\label{eq:restricted-gt-payments}
    \forall i \in N, c_i \in [0,B] : \qquad & \left(\forall \vec{c}_{-i} \in [0, B]^{n-1} : x_i(c_i, \vec{c}_{-i})=0 \wedge \sup_{\vec{c'}} p_i(\vec{c}') \leq c_i \right) \vee \notag \\
    & \left(\exists \vec{c}^*_{-i}\in [0,B]^{n-1} : x_i(c_i, \vec{c}_{-i}^*)=1 \wedge p_i(c_i, \vec{c}_{-i}^*) = \sup_{\vec{c}'} p_i(\vec{c}')\right) .
\end{align}
\end{definition}

\begin{restatable}{lemma}{bnomchar}
\label{lem:bnomchar}
    A mechanism $\mech$ that satisfies normalized payments is BNOM if and only if $\mech$ uses restricted golden ticket payments.
\end{restatable}

\begin{proof}[Proof (Lemma~\ref{lem:bnomchar})]
    ($\Rightarrow$) Suppose that $\mech$ satisfies normalized payments (NP) and BNOM. Let $i \in N$ and $c_i \in [0,B]$. We will show that the condition inside the first two quantifiers of (\ref{eq:restricted-gt-payments}) holds for $i$ and $c_i$. The statement obviously holds if $x_i(\vec{c}) = 0$ for all $\vec{c}$, therefore, we assume from here on that there exists a $\vec{c}$ for which $x_i(\vec{c})=1$. Suppose for contradiction that the condition does not hold, i.e.,
    \begin{equation}\label{eq:negation}
        \begin{split} 
        & \left(\exists \vec{c}_{-i} \in [0,B]^{n-1} : x_i(c_i, \vec{c}_{-i}) = 1 \vee \sup_{\vec{c'}} p_i(\vec{c'}) > c_i \right) \wedge \\
        & \left( \forall \vec{c}_{-i}^* \in [0,B]^{n-1} : x_i(c_i,\vec{c}_{-i}^*) = 0 \vee p_i(c_i,\vec{c}_{-i}^*) < \sup_{\vec{c'}} p_i(\vec{c'})\right) .
        \end{split}
    \end{equation}
    We split our analysis into two cases and consider first the case that $i$ never gets hired when declaring $c_i$. In that case, the above expression simplifies to $\sup_{\vec{c'}} p_i(\vec{c'}) > c_i$ .
    This implies 
    \begin{align*}
        & \sup_{\vec{c'}} u_i^{c_i}(\vec{c'}) = \sup_{\vec{c'}} (p_i(\vec{c'}) - c_i x_i(\vec{c'})) = \sup_{\vec{c'}} (p_i(\vec{c'}) - c_i) > 0 = \sup_{\vec{c}_{-i}} u_i^{c_i}(c_i,\vec{c}_{-i}) 
    \end{align*}
    and contradicts BNOM. Note that for the second inequality, we used NP along with the fact that there exists a $\vec{c}$ for which $x_i(\vec{c})=1$.

    Next, we consider the case that there exists a bid profile $\vec{c}_{-i}$ such that $x_i(c_i,\vec{c}_{-i}) = 1$. Define $S = \{\vec{c}_{-i} \in [0,B]^{n-1}\ :\ x_i(c_i,\vec{c}_{-i}) = 1\}$. Now, (\ref{eq:negation}) simplifies to $\forall \vec{c}_{-i}^* \in S : p_i(c_i,\vec{c}_{-i}^*) < \sup_{\vec{c'}} p_i(\vec{c'})$. The right-hand side of the latter inequality is positive by the fact that payments are non-negative, so by NP we obtain $\sup_{\vec{c}_{-i}} p_i(c_i,\vec{c}_{-i}) < \sup_{\vec{c'}} p_i(\vec{c'})$, and therefore we arrive at a contradiction to BNOM:
    \begin{align*}
    \sup_{\vec{c}'} u_i^{c_i}(\vec{c'}) = \sup_{\vec{c}'} p_i(\vec{c'}) - c_i > \sup_{\vec{c}_{-i}} p_i(c_i,\vec{c}_{-i}) - c_i = \sup_{\vec{c}_{-i}} u_i^{c_i}(c_i,\vec{c}_{-i}).
    \end{align*}

    ($\Leftarrow$) We suppose $\mech$ satisfies NP and restricted golden ticket payments. We will derive that $\mech$ also satisfies BNOM. Let $i \in N$ be an agent and let $c_i \in [0,B]$ be any type for $i$. We split the proof again into two cases. The first case is where $\mech$ never assigns the item to $i$ when they declare $c_i$. In this case, we see that, according to (\ref{eq:restricted-gt-payments}), $\sup_{\vec{c}'} p_i(\vec{c}') \leq c_i$.
    Thus, 
    \begin{align*}
        \sup_{\vec{c}'} u_i^{c_i}(\vec{c}') = \sup_{\vec{c}'} p_i(\vec{c}') - c_i \leq 0 = \sup_{\vec{c}_{-i}} u_i^{c_i}(c_i,\vec{c}_{-i}),
    \end{align*}
    i.e., the BNOM property holds for $i$ with type $c_i$.
    
    The second case we consider is that there is a bid profile $\vec{c}_{-i}$ such that $x_i(c_i,\vec{c}_{-i}) = 1$. Now, by (\ref{eq:restricted-gt-payments}), it holds that there also exists a $\vec{c}_{-i}^*$ for which $x_i(c_i,\vec{c}_{-i}^*) = 1$ and $p_i(c_i,\vec{c}_{-i}^*) = \sup_{\vec{c'}} p_i(\vec{c'})$. Therefore,
    \begin{align*}
    \sup_{\vec{c'}} u_i^{c_i}(\vec{c'}) = \sup_{\vec{c'}} (p_i(\vec{c'}) - c_i) = p_i(c_i,\vec{c}_{-i}^*) - c_i \leq \sup_{\vec{c}_{-i}} (p_i(c_i,\vec{c}_{-i}^*) - c_i) = \sup_{\vec{c}_{-i}} u_i^{c_i}(c_i,\vec{c}_{-i}^*) , 
    \end{align*}
    which establishes the BNOM property for $i$ with type $c_i$.
\end{proof}

With the above lemma, we are ready to prove our characterization result for BNOM.
\begin{proof}[Proof of Proposition \ref{prop:bnomchar}]
    ($\Rightarrow$) Suppose that $\mech$ satisfies normalized payments (NP),  individual rationality (IR), and BNOM. Let $i \in N$. We will show that the formula inside the first quantifier of (\ref{eq:threshold-gt-payments}) holds for $i$. By Lemma \ref{lem:bnomchar}, (\ref{eq:restricted-gt-payments}) is true. Therefore, we will show that 
    \begin{itemize}
        \item (P1) If for any $c_i \in [0,B]$ the first line of (\ref{eq:restricted-gt-payments}) holds, then the first line of (\ref{eq:restricted-gt-payments}) also holds for all types greater than $c_i$.
    \end{itemize} 
    The claim will then follow by the following argument: (P1) implies the existence of a threshold $b_i \in [0.B]$ such that $x_i(c_i,\cdot) = 0$ if and only if $c_i > b_i$. Line 1 of (\ref{eq:restricted-gt-payments}) then implies $\sup_{\vec{c'}} p_i(\vec{c'}) \leq b_i$. Suppose now, for contradiction, that the latter inequality is strict, i.e., let $b_i' = \sup_{\vec{c'}} p_i < w_i$. Let $c_i \in (b_i',b_i)$ and observe that $\sup_{\vec{c}_{-i}} u_i^{c_i}(c_i, \vec{c}_{-i}) = \sup_{\vec{c}_{-i}} (p_i(c_i,\vec{c}_{-i}) - c_i) = b_i' - c_i < b_i - c_i < 0$, which violates IR, so yields a contradiction. Thus, it must be that $\sup_{\vec{c'}} p_i(\vec{c'}) = b_i$. Line 2 of (\ref{eq:restricted-gt-payments}) then gives us that for all $c_i < b_i$, there exists a $\vec{c}_{-i}^*$ such that $p_i(c_i,\vec{c}_{-i}^*) = \sup_{\vec{c}'} p_i(\vec{c'})$, or equivalently, $\sup_{\vec{c}_{-i}} p_i(c_i,\vec{c}_{-i}) = \sup_{\vec{c}'} p_i(\vec{c'})$. So far, this establishes the first two clauses of (\ref{eq:threshold-gt-payments}). For the third line, note that if it holds that $x_i(b_i, \cdot) = 0$, then we are done, and otherwise $\sup_{\vec{c}_{-i}} p_i(b_i,\vec{c}_{-i}) = b_i$ by (\ref{eq:restricted-gt-payments}), which yields Clause 3 of (\ref{eq:threshold-gt-payments}).

    It remains to show that (P1) holds. Suppose not, then there exist two types $c_i,c_i'$ with $c_i' < c_i$ and $\vec{c}_{-i}$ such that both $x_i(c_i,\vec{c}_{-i}) = 1$, and $x_i(c_i',\cdot) = 0$. By IR, $p_i(c_i, \vec{c}_{-i}) \geq c_i > c_i'$. This yields a violation to BNOM: $\sup_{\vec{c'}_{-i}} u_i^{c_i'}(c_i',\vec{c'}_{-i}) = 0 < p_i(c_i,\vec{c}_{-i}) - c_i' = u_i^{c_i'}(c_i,\vec{c}_{-i}) \leq \sup_{\vec{c}} u_i^{c_i'}(\vec{c})$.

    ($\Leftarrow$) Suppose that $\mech$ satisfies NP, IR, and (\ref{eq:threshold-gt-payments}). Let $i \in N$. We will show that $\mech$ is BNOM for $i$. Let $b_i \in [0,B]$ be the unique threshold for which the formula of (\ref{eq:threshold-gt-payments}) inside the first two quantifiers holds. Let $c_i \in [0,B]$ be any type for bidder $i$. 
    
    If $c_i > b_i$, then $x_i(c_i,\cdot) = 0$. For all $\vec{c'}$ for which $x_i(\vec{c'})=1$, we can bound the payment as $p_i(\vec{c}') \leq b_i < c_i$. Thus, we can derive $\sup_{\vec{c}_{-i}} u_i^{c_i}(c_i,\vec{c}_{-i}) = 0 > p_i(\vec{c'}) - c_i = u_i^{c_i}(\vec{c'})$, and hence, $\sup_{\vec{c}_{-i}} u_i^{c_i}(c_i,\vec{c}_{-i}) \geq \sup_{\vec{c'}} u_i^{c_i}(\vec{c'})$, which means that BNOM holds for $i$ when $i$'s type is greater than $b_i$.
    
    If $c_i < b_i$, then (\ref{eq:threshold-gt-payments}) tells us that $\sup_{\vec{c}_{-i}} u_i^{c_i}(c_i,\vec{c}_{-i}) = b_i - c_i$. Furthermore, also by (\ref{eq:threshold-gt-payments}), for any profile $\vec{c'}$ where $x_i(\vec{c'}) = 1$, we have $p_i(\vec{c'}) \leq b_i$, so $u_i^{c_i}(\vec{c'}) \leq b_i - c_i$. For any profile $\vec{c'}$ where $x_i(\vec{c'}) = 0$, NP gives us $u_i^{c_i}(\vec{c'}) = 0$. Therefore, $\sup_{\vec{c}_{-i}} u_i^{c_i}(c_i,\vec{c}_{-i}) = b_i - c_i \geq \sup_{\vec{c'}} u_i^{c_i}(\vec{c'})$. That is, BNOM holds for $i$ when $i$'s type is less than $b_i$.

    Lastly, for the type $b_i$, if $x_i(b_i, \cdot) = 0$ then the BNOM condition for $i$ and $b_i$ follows from a similar analysis as for the case where $c_i > b_i$, and otherwise the BNOM condition for $i$ and $b_i$ follows by reasoning analogously to the case where $c_i < w_i$.
    \end{proof}

\smallskip
The following theorem shows that, among NP, IR, BF, and BNOM mechanisms, \textsc{MaxOr}\wwm\ is in fact optimal with respect to the approximation factor for additive, monotone submodular and monotone subadditive valuations: No NP, IR, BF, BNOM mechanism can achieve an approximation factor better than $2$.

\begin{theorem}\label{thm:gt-simple}
    Let $\mech$ be any deterministic mechanism that satisfies NP, IR, BF and BNOM. 
    Then $\mech$ is not $(2-\varepsilon)$-approximate for any $\varepsilon > 0$, for additive valuation functions. 
\end{theorem}
\begin{proof}
    Suppose for contradiction that $\mech$ is $(2-\varepsilon)$-approximate for some $\varepsilon > 0$. Consider an instance with two agents with values $v_1 = v_2 = 1$. Clearly, $\mech$ should select both agents if their declared costs $(c_1,c_2)$ satisfy $c_1+c_2 \leq B$. By Proposition \ref{prop:bnomchar}, $\mech$ uses threshold golden tickets. Let $b_1$ and $b_2$ be the thresholds associated to the two agents. If for some $i \in \{1,2\}$ it holds that $b_i < B$, then Agent $i$ doesn't get selected on any profile $(c_1,c_2)$ with $c_i > b_i$ and $c_1+c_2 \leq B$, and this contradicts the fact that $\mech$ achieves a $(2-\varepsilon)$-approximation. Thus, the thresholds $b_1$ and $b_2$ are both $B$.

    Hence, if Agent $1$ declares $0$ there is a declaration $c_{2}$ such that $p_i(0,c_{2}) = B$. By BF, $p_{2}(0,c_{2}) = 0$ and by IR, $c_{2} = 0$. Thus, $p_1(0,0) = B$. By the same reasoning, there is a declaration $c_1$ such that $p_{2}(c_1,c_{2}) = p_{2}(c_1,0) = B$. By BF and IR we have again that $p_1(c_1,c_2) = p_1(0,0) =0$, which contradicts $p_1(0,0) = B$ and refutes our assumption that $\mech$ is $(2-\varepsilon)$-approximate.
\end{proof}

\subsection{WNOM Characterization and Approximation Bound of $\varphi$}\label{sec:wnomchar}
For WNOM, we provide a characterization in terms of thresholds that is similar to our BNOM characterization given in Section \ref{sec:bnomchar}. In words, our characterization property for WNOM states for each agent $i \in N$ there is a particular threshold $w_i \in [0,B]$ such that:
(1)    when declaring a type strictly above $w_i$, there exists a bid profile of the remaining agents under which $i$ is not selected by the mechanism,
(2)     declaring a type strictly below $w_i$ guarantees $i$ to get selected, and the minimum payment received by $i$ is $w_i$,
(3)     when declaring type $w_i$, either of the above two cases apply.
\begin{definition}\label{def:threshold-ws}
    A mechanism $\mech$ uses \emph{threshold wooden spoons} iff 
    \begin{align}
    \forall i \in N : \exists w_i \in [0,B] : & \quad 
    \left(\rule{0ex}{3ex} \forall c_i \in (w_i,B] : \inf_{\vec{c}_{-i}} x_i(c_i,\vec{c}_{-i}) = 0\right) \wedge \notag  
    \left(\forall c_i \in [0,w_i) : \inf_{\vec{c}_{-i}} p_i(c_i,\vec{c}_{-i}) = w_i \right) \notag \\
    & \wedge \, \left(\inf_{\vec{c}_{-i}} p_i(w_i,\vec{c}_{-i}) = w_i \vee \inf_{\vec{c}_{-i}} x_i(w_i,\vec{c}_{-i}) = 0\right) \label{eq:threshold-ws} .
\end{align}
\end{definition}

\begin{restatable}{proposition}{propwnomchar}
\label{prop:wnomchar}
A mechanism $\mech$ that satisfies NP and IR is WNOM if and only if $\mech$ uses threshold wooden spoons.
\end{restatable}
\begin{proof}
    ($\Rightarrow$) Suppose that $\mech$ satisfies NP, IR and WNOM, and let $i \in N$. We will show that the formula inside the first quantifier of (\ref{eq:threshold-ws}) holds for $i$. Suppose that for some $c_i \in [0,B]$, Agent $i$ always gets selected by the mechanism when declaring $c_i$, i.e., $x_i(c_i,\cdot) = 1$. Let $c_i' \in [0, c_i)$ be any declaration less than $c_i$. If there exists a $\vec{c}_{-i}'$ such that $x_i(c_i',\vec{c}_{-i}') = 0$, then 
\begin{equation*}
\inf_{c_{-i}''} u_i^{c_i'}(c_i',\vec{c}_{-i}'') = 0 < \inf_{\vec{c}_{-i}''}p_i(c_i,\vec{c}_{-i}'') - c_i' = \inf_{\vec{c}_{-i}''}(p_i(c_i,\vec{c}_{-i}'') - x_i(c_i,\vec{c}_{-i}'')c_i') = \inf_{\vec{c}_{-i}''} u_i^{c_i'}(c_i,\vec{c}_{-i}''),
\end{equation*}
where the first equality holds by NP, and the inequality holds because $p_i(c_i,\vec{c}_{-i}'') \geq c_i \geq c_i'$, by IR. This contradicts WNOM, and therefore it must hold that $x_i(c_i', \cdot) = 1$ for all $c_i' < c_i$. This implies the existence of a threshold $w_i$ such that $x_i(c_i',\cdot) = 1$ if $c_i' < w_i$, and such that $\inf_{\vec{c}_{-i}''} x_i(c_i',\vec{c}_{-i}'') = 0$ if $c_i' > w_i$. Note that the latter establishes the first clause of (\ref{eq:threshold-ws}).

Let $c_i \in [0,w_i)$ now be any declaration below the threshold, so that $i$ is guaranteed to be selected by $\mech$ when declaring $c_i$. Let $\check{p}_i = \inf_{\vec{c}_{-i}} p_i(c_i,\vec{c}_{-i})$, and suppose for contradiction that $\check{p}_i < w_i$. By declaring any $c_i'$ strictly in between $\check{p}_i$ and $w_i$, Agent $i$ is guaranteed to be selected and receives a payment of at least $c_i'$ (which must hold by IR). Hence,
\begin{equation*}
\inf_{\vec{c}_{-i}} u_{i}^{c_i}(c_i,\vec{c}_{-i}) = \check{p_i} - c_i < c_i' - c_i \leq \inf_{\vec{c}_{-i}} p_i(c_i',\vec{c}_{-i}) - c_i = \inf_{\vec{c}_{-i}} u_i^{c_i}(c_i,\vec{c}_{-i}), 
\end{equation*}
which contradicts WNOM. Thus, we conclude that $\check{p}_i \geq w_i$. Next, suppose for the sake of contradiction that $\check{p}_i > w_i$. Let $c_i'$ again be any declaration strictly in between $w_i$ and $\check{p}_i$. We now see that 
\begin{equation*}
\inf_{\vec{c}_{-i}} u_{i}^{c_i'}(c_i',\vec{c}_{-i}) = 0 < \check{p}_i - c_i' = \inf_{\vec{c}_{-i}} p_i(c_i,\vec{c}_{-i}) - c_i' = \inf_{\vec{c}_{-i}} u_i^{c_i'}(c_i, \vec{c}_{-i}), 
\end{equation*}
which is again a contradiction to WNOM and yields that $\inf_{\vec{c}_{-i}} p_i(c_i,\vec{c}_{-i}) = \check{p}_i = w_i$. This establishes the second clause of (\ref{eq:threshold-ws}).

For the third clause of (\ref{eq:threshold-ws}): If $x_i(w_i,\cdot) = 1$ then (\ref{eq:threshold-ws}) holds trivially. Otherwise, it holds that $\displaystyle\inf_{\vec{c}_{-i}} p_i(w_i,\vec{c}_{-i}) \!\geq\! w_i$, by IR. If the latter inequality would be strict, we could again take any $c_i'$ strictly in between $w_i$ and $p_i(w_i,\vec{c}_{-i})$ and reach a contradiction to WNOM as follows.
\begin{equation*}
\inf_{\vec{c}_{-i}} u_{i}^{c_i'}(c_i',\vec{c}_{-i}) = 0 < \inf_{\vec{c}_{-i}} p_i(w_i,\vec{c}_{-i}) - c_i' = \inf_{\vec{c}_{-i}} u_i^{c_i'}(w_i, \vec{c}_{-i}) ,
\end{equation*}
which establishes that the third clause of (\ref{eq:threshold-ws}) holds.

($\Leftarrow$) Suppose that $\mech$ satisfies NP and IR, and uses threshold wooden spoons. Let $i \in N$ and let $c_i \in [0,B]$. We will show that the WNOM condition holds for $i$ when $c_i$ is $i$'s true cost. Let $w_i$ be the threshold for which the formula inside the first quantifier of (\ref{eq:threshold-ws}) holds. We distinguish two cases: the first case we will analyse is when $c_i < w_i$ or $c_i= w_i \wedge \inf_{\vec{c}_{-i}} p_i(w_i, \vec{c}_{-i}) = w_i$. If the first case does not hold, then according to (\ref{eq:threshold-ws}) it must be that $c_i > w_i$ or $c_i= w_i \wedge \inf_{\vec{c}_{-i}} x_i(w_i, \vec{c}_{-i}) = 0$, which comprises the second case in our analysis.

\paragraph{Case 1.} $c_i < w_i$ or $c_i = w_i \wedge \inf_{\vec{c}_{-i}} p_i(w_i, \vec{c}_{-i}) = w_i$. Then, $x_i(c_i,\cdot) = 1$ and $\inf_{\vec{c}_{-i}} p_i(c_i,\vec{c}_{-i}) = w_i$, so $\inf_{\vec{c}_{-i}} u_i^{c_i}(c_i,\vec{c}_{-i}) = w_i - c_i$. For any declaration $c_i'$ different from $c_i$ we may distinguish between two subcases: 
\begin{itemize}\itemsep0pt
    \item If $c_i' < w_i$ or $c_i' = w_i \wedge \inf_{\vec{c}_{-i}} p_i(w_i,\vec{c}_{-i}) = w_i$, then $\inf_{\vec{c}_{-i}} u_i^{c_i}(c_i',\vec{c}_{-i}) = w_i - c_i = \inf_{\vec{c}_{-i}} u_i^{c_i}(c_i,\vec{c}_{-i})$;
    \item if $c_i' > w_i$ or $c_i' = w_i \wedge \inf_{\vec{c}_{-i}} x_i(w_i,\vec{c}_{-i}) = 0$, then $\inf_{\vec{c}_{-i}} u_i^{c_i}(c_i',\vec{c}_{-i}) = 0 \leq \inf_{\vec{c}_{-i}} u_i^{c_i}(c_i,\vec{c}_{-i})$.
\end{itemize}
Hence, in both subcases, the WNOM condition holds for Agent $i$ with true cost $c_i$.

\paragraph{Case 2.} $c_i > w_i$ or $c_i = w_i \wedge \inf_{\vec{c}_{-i}} x_i(w_i,\vec{c}_{-i}) = 0$. Then, $\inf_{\vec{c}_{-i}} x_i(c_i,\vec{c}_{-i}) = 0$.  For any declaration $c_i'$ different from $c_i$ we may again distinguish between two subcases: 
\begin{itemize}\itemsep0pt
    \item If $c_i' < w_i$ or $c_i' = w_i \wedge \inf_{\vec{c}_{-i}} p_i(w_i,\vec{c}_{-i}) = w_i$, then $\inf_{\vec{c}_{-i}} u_i^{c_i}(c_i',\vec{c}_{-i}) = w_i - c_i \leq 0 = \inf_{\vec{c}_{-i}} u_i^{c_i}(c_i,\vec{c}_{-i})$;
    \item if $c_i' > w_i$ or $c_i' = w_i \wedge \inf_{\vec{c}_{-i}} x_i(w_i,\vec{c}_{-i}) = 0$, then $\inf_{\vec{c}_{-i}} u_i^{c_i}(c_i',\vec{c}_{-i}) \leq 0 = \inf_{\vec{c}_{-i}} u_i^{c_i}(c_i,\vec{c}_{-i})$.
\end{itemize}
Hence, in both subcases, the WNOM condition holds for Agent $i$ with true cost $c_i$.
\end{proof}

We can use the above characterization to establish a simple lower bound of $\varphi$ on the approximation factor achievable by NP, IR, BF, and BNOM mechanisms. 
\begin{theorem}\label{thm:ws-simple}
    Let $\mech$ be any deterministic mechanism that satisfies NP, IR, BF and WNOM. 
    Then $\mech$ is not $(\varphi-\varepsilon)$-approximate for any $\varepsilon > 0$, for additive valuation functions. Here, $\varphi = (1 + \sqrt{5})/2$ is the golden ratio. 
\end{theorem}
\begin{proof}
    Suppose for contradiction that $\mech$ is $(\varphi - \varepsilon)$-approximate for some $\varepsilon > 0$. Consider an instance with two agents with values $v_1 = \varphi$ and $v_2 = 1$. Mechanism $\mech$ should select both agents for all cost profiles $(c_1,c_2)$ that satisfy $c_1 + c_2 \leq B$, otherwise the valuation of the selected agent would be at most $\varphi$, which yields an approximation factor of at least $(\varphi + 1)/\varphi = \varphi$, and hence $\mech$ would not be $(\varphi - \varepsilon)$-approximate. On all remaining cost profiles, for which $c_1 + c_2 > B$, mechanism $\mech$ must select Agent $1$, as otherwise the approximation factor would be $\varphi/1$ so again $\mech$ would not be $(\varphi - \varepsilon)$-approximate. Thus, Agent 1 always gets selected, regardless of the declared cost profile.
    
    By Proposition \ref{prop:wnomchar}, $\mech$ uses threshold wooden spoons, and by the fact that Agent $1$ is always selected by $\mech$, regardless of their declared cost, we know that the threshold $w_1$ associated to agent 1 is $B$. This implies that Agent $1$ is always paid an amount of $B$, regardless of their reported cost. So on any profile $(c_1,c_2)$, where $c_1 + c_2 \leq B$,  and $c_2 > 0$, mechanism $\mech$ selects only agent 1 and pays them $B$, while agent 2 cannot be selected due to BF. This yields an approximation factor of $(\varphi+1)/\varphi = \varphi$ on such profiles, and contradicts that $\mech$ is $(\varphi - \varepsilon)$-approximate.
\end{proof}

\minisec{Optimal $\varphi$-Approximation Mechanism for WNOM}
The approximation factor lower bound of $\varphi$ turns out to be tight, and this is even the case for monotone subadditive valuation functions. In this section, we present the \textsc{GoldenMechanism} (Mechanism \ref{mech:wnom}), which we show to be WNOM, IR, NP, and BF, achieving an approximation guarantee of $\varphi$ for a monotone subadditive valuation function, for an arbitrary number of agents.
To define the \textsc{GoldenMechanism}, we need the following optimality notions.
\begin{definition}\label{def:goldenopt}
    Given a set of agents $N$, valuation function $V$, budget $B$, and declared costs $\vec{c}$ let 
    \begin{equation*}
    X^*(\vec{c},B) \in \arg_S \max\left\{V(S)\ :\  S \subseteq N  \wedge \sum_{i \in S} c_i \leq B\right\} 
    \end{equation*}
    be a subset of agents that is budget feasible with respect to $B$ and optimizes $V$. Furthermore, let
    \begin{equation*}
    X^*_{\geq 2}(\vec{c}_{-1},B) \in \arg_S \max\left\{V(S)\ :\  S \subseteq N\setminus\{1\}  \wedge \sum_{i \in S} c_i \leq B\right\} 
    \end{equation*}
    be a subset of agents that is budget feasible with respect to $B$ and optimizes $V$ restricted to the subsets of $N\setminus\{1\}$ (i.e., excluding Agent $1$). In case of ties (i.e., in case there is more than one maximum element in the sets in the above $\arg \max$ expressions), the maximum agent set is chosen with respect to any fixed strict total order over $2^N$ that prefers higher-cardinality sets over lower-cardinality sets.
\end{definition}
The \textsc{GoldenMechanism} is derived from an allocation rule $X1$ that, assuming $V(\{1\}) \geq \cdots \geq V(\{n\})$, selects the possibly sub-optimal set of agents $X^*_{\geq 2}(\vec{c}_{-1},B)$ whenever that choice results in an approximation factor of at most $\varphi$, and selects $X^*(\vec{c},B)$ otherwise.
\begin{definition}\label{def:X1}
    Given a set of agents $N$, valuation function $V$, budget $B$, and declared costs $\vec{c}$ let 
    \begin{equation*}
        X1(\vec{c},B) = 
        \begin{cases} X^*_{\geq 2}(\vec{c}_{-1},B) & \text{ if } V(X^*(\vec{c},B))/V(X^*_{\geq 2}(\vec{c}_{-1},B)) < \varphi , \\
        X^*(\vec{c},B) & \text{ otherwise.}
        \end{cases}
    \end{equation*}
\end{definition}
We first prove that $X1$ is monotone in Agent 1's declared cost $c_1$, which is a property that we need in order to properly understand the behaviour of the \textsc{GoldenMechanism}. 
\begin{lemma}\label{lem:x1monotone}
    Let $c_1, c_1' \in [0,B]$, $c_1' < c_1$, and let $\vec{c}_{-1} \in [0,B]^{n-1}$.
    If $1 \in X1((c_1,\vec{c}_{-1}),B)$ then $1 \in X1((c_1', \vec{c}_{-1}),B)$.
\end{lemma}
\begin{proof}
    Since $1 \in X1((c_1,\vec{c}_{-1}),B)$, it holds that $X1((c_1,\vec{c}_{-1}),B) = X^*((c_1,\vec{c}_{-1}),B)$. If we suppose that $1 \not\in X1((c_1',\vec{c}_{-1}),B)$, then  $X1((c_1',\vec{c}_{-1}),B) = X^*_{\geq 2}(\vec{c}_{-1},B)$, then we obtain that $\varphi \cdot V(X^*_{\geq 2}(\vec{c}_{-1},B)) \geq V(X^*((c_1',\vec{c}_{-1}),B)) \geq V(X^*((c_1,\vec{c}_{-1}),B))$, which implies that $X1((c_1,\vec{c}_{-1}),B) = X^*_{\geq 2}(\vec{c}_{-1},B)$, a contradiction.
\end{proof}
 
There exists a certain maximal threshold $w_1 \in [0,B]$ (where possibly $w_1 = 0$) such that $X1$ guarantees that Agent $1$ is selected whenever Agent 1 reports a cost less than $w_1$. That is, $1 \in X1((c_1, \vec{c_{-1}}),B)$ for all $c_1 \in [0,w_1)$ and $\vec{c}_{-1} \in [0,B]^{n-1}$.
Another property we need in order to well-define the \textsc{GoldenMechanism}, is that if Agent 1 declares cost $w_1$ itself, then Agent $1$ also is guaranteed to get selected. This is captured by the following lemma.
\begin{lemma}\label{lem:w1maxexists}
 If the set $Y = \{d_1\ |\ \forall \vec{d}_{-1} \in [0,B]^{n-1} : 1 \in X1(d_1,\vec{d}_{-1},B)\}$ is non-empty, then it has a maximum element.
\end{lemma}
\begin{proof}
Suppose that $Y$ is non-empty and let $w_1 = \sup Y$.
We prove this by first showing that, for all $\vec{c}_{-1}$, we have $1 \in X^*((w_1, \vec{c}_{-1}),B)$, and then proving that $X1((w_1,\vec{c}_{-1}),B)$ chooses $X^*_{\geq 2}(\vec{c}_{-1},B)$ among the two alternatives $X^*((w_1, \vec{c}_{-1}),B)$ and $X^*_{\geq 2}(\vec{c}_{-1},B)$.

First, we observe that if $w_1 = 0$, then the statement is trivial: Since $Y$ is non-empty, $Y = \{0\}$. Thus, from here onward we assume $w_1 > 0$. Let $\vec{c}_{-1} \in [0,B]^{n-1}$ be arbitrary. The value $V(X^*((c_1,\vec{c}_{-1}),B))$, as a function of $c_1$, is monotonically non-increasing. Combining this with the fact that the number of subsets of $N$ is finite, there is an $\epsilon > 0$ such that $X^*((c_1, \vec{c}_{-1}),B)$ (again as a function of $c_1$) is constant on $c_1 \in [w_1 - \epsilon, w_1)$. Let $S$ be the set of agents that $X^*((c_1, \vec{c}_{-1}),B)$ maps to on the domain $c_1 \in [w_1 - \epsilon, w_1)$. By the definition of $w_1$ and Lemma \ref{lem:x1monotone}, Agent $1$ is in $S$. Furthermore for all $c_1 \in [w_1 - \epsilon, w_1)$,
$
c_1 + \sum_{i \in S\setminus\{1\}} c_i \leq B ,    
$
and therefore 
\begin{equation*}
    w_1 + \sum_{i \in S \setminus \{1\}} c_i = \lim_{c_1 \uparrow w_1} \left[ c_1 + \sum_{i \in S \setminus \{1\}} c_i\right] \leq B.
\end{equation*}
This establishes that $S$ is a budget-feasible set under the cost profile $(w_1,\vec{c}_{-1})$
Suppose now, for contradiction, that $X^*((w_1,\vec{c}_{-1}),B) = T \not= S$. In that case, by the non-increasingness of  $V(X^*((c_1,\vec{c}_{-1}),B))$ in the argument $c_1$, it holds that $V(X^*((w_1,\vec{c}_{-1}),B)) \leq V(S)$. If the latter inequality is strict, there is a contradiction with $S$ being budget-feasible under $(w_1,\vec{c}_{-1})$: $X^*((w_1, \vec{c}_{-1}),B)$ should select $S$ instead of $T$.  If on the other hand, the latter inequality holds with equality, there is a contradiction with the tie-breaking rule (see Definition \ref{def:goldenopt}, which prescribes that $S$ gets priority over $T$, and hence $X^*((w_1, \vec{c}_{-1}),B)$ should select $S$ rather than $T$. Altogether, we conclude that $X^*((w_1,\vec{c}_{-1}),B) = S$, and hence that $1 \in X^*((w_1,\vec{c}_{-1}),B)$, concluding the first part of the proof.

For the second part of the proof, i.e., showing that $X1((w_1,\vec{c}_{-1}),B) \not= X^*_{\geq 2}(\vec{c}_{-1},B)$, suppose for contradiction that  $X1((w_1,\vec{c}_{-1}),B) = X^*_{\geq 2}(\vec{c}_{-1},B)$. 
By definition of $X1$ it holds that $V(S) = V(X^*((w_1,\vec{c}_{-1}),B) < \varphi 
X^*_{\geq 2}((w_1,\vec{c}_{-1}),B)$. Now, let $c_1 \in [w_1 - \epsilon, w_1)$. Since  $X^*((c_1,\vec{c}_{-1}),B) = S$ we obtain that $V(X^*((c_1,\vec{c}_{-1}),B) = V(S) < \varphi X^*_{\geq 2}(\vec{c}_{-1},B)$, so that $X1(c_1, \vec{c}_{-1},B) = X^*_{\geq 2}(\vec{c}_{-1},B)$, which contradicts that by definition of $w_1$ it must hold that $1 \in X1((c_1,\vec{c}_{-1}),B)$.
\end{proof}
With the above definitions and ideas in mind, the \textsc{GoldenMechanism} is now straightforward to define: On reported cost profile $\vec{c}$ we let the \emph{proxy cost profile} $\vec{c}'$ be such that $c_1' = \max\{w_1, c_1\}$ and $c_i' = c_i$ for all $i > 2$. The mechanism then selects the set $X1(\vec{c}',B)$ and pays each selected agent $i$ an amount of $\vec{c}_i'$. In other words, we use the selection rule $X1$ and pay the winning agents their declared cost, where for Agent $1$ the mechanism ``pretends'' that the declared cost is $w_1$ in case they declared less than $w_1$. The mechanism makes an exception when the reported costs of all Agents except Agent 2 are equal to $B$, in which case the optimum agent set $X^*(\vec{c},B)$ is always selected (which is $\{1\}$ if $c_2 > 0$ and $\{1,2\}$ if $c_2 = 0$).

See Mechanism \ref{mech:wnom} for a precise description of the \textsc{GoldenMechanism}.
\begin{mechanism}[t]
\caption{$\textsc{GoldenMechanism}(I)$\label{mech:wnom}}

\nonl \hspace*{-1em} $\rhd$ {\bf{Input:}}  
Instance $I=(N, \vec{c}, V, B)$, with $N = \{1,\ldots,n\}$, and $V$ monotone subadditive.\;

\nonl \hspace*{-1em} $\rhd$ Let $X1$ be as in Definition \ref{def:X1}.\\
\nonl \hspace*{-1em} Rename the agents such that $V(\{1\} \geq V(\{2\}) \geq \cdots \geq V(\{n\})$.\\
\nonl \hspace*{-1em} $\rhd$ Let $Y = \{d_1\ |\ \forall \vec{d}_{-1} \in [0,B]^{n-1} : 1 \in X1((d_1,\vec{d}_{-1}),B)\}$, as in Lemma \ref{lem:w1maxexists}.\\
\If{$c_i = B$ for all $i \in N \setminus \{2\}$ \label{line:allBs:start}}{Set $x_1 = 1$ \\
Set $x_2 = 0$ if $c_2 > 0$ and set $x_2 = 1$ otherwise. \\
Set $x_i = 0$ for all $i \in N \setminus \{1,2\}$ \\
Set $p_i = x_ic_i$ for all $i \in N$ \\
\Return $(\vec{x},\vec{p})$ \label{line:allBs:end}}
\If{$Y \not= \varnothing$\label{mech:wnom:w1start}}{Let $w_1 = \max Y$\algcomf{The max exists by Lemma \ref{lem:w1maxexists}.}}
\Else{Let $w_1 = 0$\label{mech:wnom:w1end}} 
Let $\vec{c}' = (\max\{w_1,c_1\},c_2, c_3, \ldots, c_m)$ \label{line:cprime} \\
Set $x_i = 1$ for all $i \in X1(\vec{c}',B)$ \\
Set $x_i = 0$ for all $i \in N \setminus X1(\vec{c}',B)$ \\
Set $p_i = x_ic_i'$ for all $i \in N$ \\
\Return $(\vec{x}, \vec{p})$
\end{mechanism}

\begin{restatable}{theorem}{thmGoldMech}
\label{thm:GoldMech}
\textsc{GoldenMechanism} is BF, IR, WNOM and $\varphi$-approximate for agents with monotone subadditive valuation functions, where $\varphi = \frac{1+\sqrt{5}}{2}$ is the golden ratio. 
\end{restatable}

\begin{proof}
The IR property follows straightforwardly from the fact that all selected agents get paid their declared cost, and all other agents get paid 0.

The BF property follows from the fact that the selection rule $X1$ (Definition \ref{def:X1}) always picks a set of agents for which the sum of declared costs does not exceed $B$.

For WNOM, we prove that for each $i \in N$ there exists a threshold $w_i \in [0,B]$ such that all three clauses of (\ref{eq:threshold-ws}) hold, which implies that the \textsc{GoldenMechanism} uses threshold wooden spoons and is thus WNOM by Proposition \ref{prop:wnomchar} (as the \textsc{GoldenMechanism} trivially also satisfies NP). 

For Agent $1$, the threshold $w_1$ set in Lines \ref{mech:wnom:w1start}-\ref{mech:wnom:w1end} of Mechanism \ref{mech:wnom} satisfies the three clauses of (\ref{eq:threshold-ws}): Suppose first that if Agent 1 declares $B$, there exists a $\vec{c}_{-1}$ such that $x_1(B,\vec{c}_{-1}) = 0$. In that case, for any declared cost $c_1$ that lies above $w_1$, by the definition of $w_1$ there exists a profile $\vec{c}_{-1}$ for which $x_i(c_1,\vec{c}_{-1}) = 0$, which means that the first clause of (\ref{eq:threshold-ws}) holds. For any declared cost $c_1$ in $[0,w_1]$, the mechanism always selects Agent $1$ (as per $w_1$'s definition) and pays Agent 1 an amount of $w_1$, satisfying Clauses 2 and 3 of (\ref{eq:threshold-ws}).
On the other hand suppose that $x_1(B,\cdot) = 1$. This implies that Agent $1$ is selected by the mechanism even if all other players bid a profile $\vec{c}_{-1}$ that is positive in each coordinate and such that $N\setminus\{1\}$ is budget-feasible. Hence we have $\varphi X^*_{\geq 2}(\cdot, B) = \varphi V(N\setminus\{1\}) < V(\{1\}) \leq X^*(\cdot, B)$, which means that $X1$ always selects Agent $1$ and that $w_1 = B$. Therefore, Clause 1 of (\ref{eq:threshold-ws}) is satisfied trivially. Furthermore, for any declared cost $c_1 \in [0,B]$, by Lemma \ref{lem:x1monotone}, the mechanism always selects Agent $1$ and pays Agent 1 an amount of $w_1$, satisfying Clauses 2 and 3 of (\ref{eq:threshold-ws}).

For each other agent $i \in N \setminus \{1\}$, the threshold $w_i = 0$ satisfies the three clauses: For any declaration $c_i > 0$, consider the profile $\vec{c}_{-i} = (B,\ldots, B)$. Lines \ref{line:allBs:start}-\ref{line:allBs:end} of Mechanism \ref{mech:wnom} ensure that only Agent $1$ is selected when the declared profile is $(c_i,\vec{c}_{-i})$, and hence will not select agent $i$. Thus, Clause 1 of $(\ref{eq:threshold-ws})$ is satisfied. When instead declaring $c_i = 0$, if $i$ is selected, $i$ will be paid $0$ as every selected agent pays their bid in this mechanism, so Clause 3 of (\ref{eq:threshold-ws}) holds. Clause 2 of (\ref{eq:threshold-ws}) trivially holds true. This proves that the \textsc{GoldenMechanism} is WNOM.

Next, we turn to the approximation guarantee of the mechanism. If the reported cost profile satisfies that $c_i = B$ for all $i \in N\setminus\{2\}$, the mechanism outputs an optimal solution by Lines \ref{line:allBs:start}-\ref{line:allBs:end}.

Otherwise, if $c_1 > w_1$, the mechanism selects the agents $X1(\vec{c},B)$, as in Line \ref{line:cprime} it holds that $\vec{c}' = \vec{c}$. The approximation factor of $\varphi$ follows directly from from the definition of $X1$ (Definition \ref{def:X1}).

For the remaining case, it holds that $c_1 \leq w_1$, and the mechanism selects the set $X1(\vec{c}',B)$ where $c_1' = w_1$ and $c_i' = c_i$ for all $i \in N\setminus\{1\}$. From the definition of $w_1$ it follows that $1 \in X1(\vec{c}',B)$.
The set $X1(\vec{c}',B)$ is potentially sub-optimal, and hence it remains to prove that $V(X^*(\vec{c},B))/V(X1(\vec{c}',B)) \leq \varphi$.
For the numerator, the inequality $V(X^*(\vec{c},B)) \leq V(\{1\} \cup X^*_{\geq 2}(\vec{c}_{-1},B)) \leq V(\{1\}) + V(X^*_{\geq 2}(\vec{c}_{-1},B))$ holds by subadditivity, and for the denominator, we observe that $V(X1(\vec{c}',B)) = V(\{1\} \cup X^*_{\geq 2}(\vec{c}_{-1}, B - w_1))$. 
\begin{equation}\label{eq:goldenproof1}
\frac{V(X^*(\vec{c},B))}{V(X1(\vec{c}',B))} \leq \frac{V(\{1\}) + V(X^*_{\geq 2}(\vec{c}_{-1},B))}{V(\{1\} \cup X^*_{\geq 2}(\vec{c}_{-1}, B - w_1))}.
\end{equation}
Because $X1$ selects $X^*(\vec{c}',B)$ rather than $X^*_{\geq 2}(\vec{c}_{-1},B)$ when given the profile $\vec{c}'$, we obtain by the definition of $X1$ (Definition \ref{def:X1}) that $V(X^*_{\geq 2}(\vec{c}_{-1},B)) \leq V(X^*(\vec{c}',B))/\varphi = (1/\varphi)\cdot V(\{1\} \cup X^*_{\geq 2}(\vec{c}_{-1},B-w_1)))$, where the equality holds because Agent 1 is guaranteed to be selected when bidding $c_1 \leq w_1$.
Thus, from (\ref{eq:goldenproof1}) we obtain
\begin{eqnarray*}
\frac{V(X^*(\vec{c},B))}{V(X1(\vec{c}',B))} & \leq & \frac{ V(\{1\}) +
 (1/\varphi)\cdot V(\{1\} \cup X^*_{\geq 2}(\vec{c}_{-1},B-w_1)))}{V(\{1\} \cup X^*_{\geq 2}(\vec{c}_{-1}, B - w_1))} \\
& \leq & \left(1 + \frac{1}{\varphi}\right) \cdot  \frac{V(\{1\} \cup X^*_{\geq 2}(\vec{c}_{-1}, B - w_1))}{V(\{1\} \cup X^*_{\geq 2}(\vec{c}_{-1}, B - w_1))}
 \quad = \quad \varphi .
\end{eqnarray*}
\end{proof}

\section{Optimal Randomized Mechanisms}

Theorem~\ref{thm:gt-simple} demonstrates that any deterministic mechanisms satisfying (B)NOM cannot be $(2-\varepsilon)$-approximate for any $\varepsilon > 0$. 
In this section, we show that this impossibility result breaks if we allow mechanisms to be randomized.

The idea behind our randomized mechanism is 
very simple: We select a set of golden tickets and wooden spoons for each of the players completely randomly, such that the probability that a given bid profile coincides with any golden ticket and wooden spoon is $0$ (or close to $0$, in case we insist on randomizing over a finite set of mechanisms). 

Let $A$ be any algorithm for solving the packing problem (\ref{eq:opt-packing}), and let $\gamma$ be the approximation factor it achieves. 
We define our randomized mechanism $MR_A$ as follows. $MR_A$ draws numbers $GT_i \in [0,B]^{n-1}$ and $WS_i \in [0,B]^{n-1}$ for all $i \in [n]$ independently, uniformly at random. Given bid profile $\vec{c}$, the mechanism then computes and outputs a $\gamma$-approximate allocation $\vec{x}$ and uses first-price payments $\vec{p} = \vec{c} \cdot \vec{x}$, unless for some $i \in N$ it holds that $\vec{c}_{-i} = GT_i$ or $\vec{c}_{-i} = WS_i$: If $\vec{c}_{-i} = GT_i$, then only agent $i$ is selected by the mechanism and is given a payment of $B$, whereas if $\vec{c}_{-i} = WS_i$, not any player is selected by the mechanism.

\begin{theorem}\label{thm:rand-BNOM}
    The randomized mechanism $MR_A$, where $A$ is a $\gamma$-approximation algorithm for the packing problem (\ref{eq:opt-packing}), satisfies NP, IR, BF, and NOM universally (i.e., every deterministic mechanism in its support has those four properties), and is $\gamma$-approximate in expectation.
\end{theorem}

\begin{proof}
It follows from the definitions of threshold golden tickets and threshold wooden spoons (Definitions \ref{def:threshold-gt-payments} and \ref{def:threshold-ws}) that $MR_A$ is a probability distribution over deterministic mechanisms that each implement threshold wooden spoons with threshold $0$ for each agent, and threshold golden tickets with threshold $B$ for each agent. Every mechanism in the support of $MR_A$ thus satisfies NOM, and the other properties NP, IR, and BF follow directly from the definition of $MR_A$.

For the approximation guarantee, let $\vec{c}$ be a declared cost profile and let $X(\vec{c})$ be the random allocation of our mechanism. Clearly, the probability that there is an $i \in [n]$ such that $\vec{c}_{-i} = GT_i$ or $\vec{c}_{-i} = WS_i$, is $0$. Therefore, with probability $1$, mechanism $MR_A$ selects a set of agents that is within a factor $\gamma$ from optimal.
\end{proof}

Note that Theorem~\ref{thm:gt-simple} holds independently of any computational constraints (i.e., even for $\gamma = 1)$. 
Our randomized mechanism defines a probability distribution over deterministic mechanisms, each satisfying NP, IR, BF and BNOM deterministically---only the approximation guarantee holds in expectation. 

In the above presentation, the randomization is done over an uncountable set of mechanisms (i.e., one deterministic mechanism for each choice of golden tickets and wooden spoons in $[0,B]^{(n-1)\cdot 2n}$). One can obtain similar results by instead randomizing over a small finite set of golden tickets and wooden spoons: If instead of drawing uniform random vectors from $[0,B]^{(n-1) \cdot 2n}$, we can just let the mechanism choose uniformly at random from a set of $\ell$ specifications of the $2n$ golden ticket and wooden spoon vectors, where these are defined in such a way that none of the golden tickets and wooden spoons among these $\ell$ specifications coincide. Denote this modification of $RM_A$ by $RM_A'$. Now, we observe that for a given bid profile $\vec{c}$, for each agent $i$ the probability that $c_{-i} = GT_i$ or $c_{-i} = WS_i$ is at most $1/\ell$. Thus, the probability that $RM_A'$ selects a set that is within a factor $\gamma$ from optimal is $1 - n/\ell$, and hence $MR_A'$ has an approximation guarantee of $\gamma(\ell/(\ell - n))$, or equivalently $\gamma + \epsilon$ where $\epsilon = \gamma(n/(\ell-n))$. Thus, for a given $\epsilon$ one must set the number of mechanisms $\ell$ in the support of $MR_A$ as $\ell = (\gamma n + \epsilon n)/\epsilon$.

\begin{corollary}
       The randomized mechanism $MR_A'$, with $\ell$ deterministic mechanisms in its support, and where $A$ is a $\gamma$-approximation algorithm for the packing problem (\ref{eq:opt-packing}), satisfies NP, IR, BF, and NOM universally, and is $(\gamma + \epsilon)$-approximate in expectation with $\epsilon = \gamma(n/(\ell n))$. 
\end{corollary}

Note that the randomization technique we are using here is not specific to the budgeted procurement setting studied here. This simple technique is of more general interest and can be used in any mechanism design setting where each agent has a sufficiently large strategy set, to obtain randomized universal NOM mechanisms of which the solution quality is optimal or near-optimal in expectation. This highlights another aspect in which NOM mechanism design is markedly different from achieving classical dominant strategy incentive compatibility (where obtaining a good approximation factor is often challenging also when allowing randomization).

\section*{Acknowledgements}
AT was partially supported by the Gravitation Project NETWORKS, grant no.~024.002.003, and the EU Horizon 2020 Research and Innovation Program under the Marie Skłodowska-Curie Grant Agreement, grant no.~101034253.
BdK was partially supported by EPSRC grant EP/X021696/1.

\bibliographystyle{plainnat}
\bibliography{references}

\end{document}